\documentclass[a4paper,UKenglish,cleveref, autoref, thm-restate]{lipics-v2021}
\pdfoutput=1
\usepackage{hyperref}

\usepackage{amsfonts} 
\usepackage{mathtools}
\usepackage{graphicx}
\usepackage{algorithmic}
\usepackage{algorithm}
\usepackage[referable]{threeparttablex}
\usepackage{tabu}
\usepackage{multirow}
\usepackage{makecell}
\usepackage{longtable}
\usepackage{pifont}
\usepackage{xcolor}
\usepackage{colortbl}
\usepackage{booktabs}
\usepackage{blindtext}
\usepackage{caption}
\usepackage{calc}
\usepackage{todonotes}
\usepackage{pdflscape}
\definecolor{lightergray}{rgb}{0.86, 0.86, 0.86}
\usepackage[boldmath,autolanguage]{numprint}

\newif\ifDoubleBlind
\DoubleBlindfalse

\newif\ifFull
\Fullfalse

\npthousandsep{\,}
\newtheorem{eccreduction}{ECC Reduction}

\newtheorem{vccreduction}{VCC Reduction}

\title{Solving Edge Clique Cover Exactly via Synergistic Data Reduction}

\author{Anthony Hevia}{Hamilton College, Clinton, NY, USA}{}{https://orcid.org/0009-0003-2049-954X}{}%TODO mandatory, please use full name; only 1 author per \author macro; first two parameters are mandatory, other parameters can be empty. Please provide at least the name of the affiliation and the country. The full address is optional. Use additional curly braces to indicate the correct name splitting when the last name consists of multiple name parts.
\author{Benjamin Kallus}{Dartmouth College, Hanover, NH, USA}{benjamin.p.kallus.gr@dartmouth.edu}{}{}
\author{Summer McClintic}{Hamilton College, Clinton, NY, USA}{}{https://orcid.org/0009-0007-3088-5451}{}
\author{Samantha Reisner}{Hamilton College, Clinton, NY, USA}{}{}{}
\author{Darren Strash\footnote{Corresponding author.}}{Hamilton College, Clinton, NY, USA}{dstrash@hamilton.edu}{https://orcid.org/0000-0001-7095-8749}{}
\author{Johnathan Wilson}{Hamilton College, Clinton, NY, USA}{}{https://orcid.org/0009-0006-2992-8840}{}

\authorrunning{A. Hevia, B. Kallus, S. McClintic, S. Reisner, D. Strash, J. Wilson} %TODO mandatory. First: Use abbreviated first/middle names. Second (only in severe cases): Use first author plus 'et al.'

\Copyright{Anthony Hevia, Benjamin Kallus, Summer McClintic, Samantha Reisner, Darren Strash, and Johnathan Wilson} %TODO mandatory, please use full first names. LIPIcs license is "CC-BY";  http://creativecommons.org/licenses/by/3.0/


%TODO mandatory; please add comma-separated list of keywords

\category{} %optional, e.g. invited paper

\acknowledgements{We thank the anonymous reviewers for their insightful feedback, David Swartz from Hamilton College for technical support, and Adam Chrisman, Caitlin Matwijec-Walda, and Jon Matwijec-Walda for a cozy space to work at The Copper Easel and Superofficial in Rome, NY.}%optional

\nolinenumbers %uncomment to disable line numbering

\EventEditors{Inge Li G{\o}rtz, Martin Farach-Colton, Simon J. Puglisi, and Grzegorz Herman}
\EventNoEds{4}
\EventLongTitle{31st Annual European Symposium on Algorithms (ESA 2023)}
\EventShortTitle{ESA 2023}
\EventAcronym{ESA}
\EventYear{2023}
\EventDate{September 4--6, 2023}
\EventLocation{Amsterdam, the Netherlands}
\EventLogo{}
\SeriesVolume{274}
\ArticleNo{104}
\begin{document}

\hideLIPIcs

\maketitle

\begin{abstract}
%\small\baselineskip=9pt 
The edge clique cover (ECC) problem---where the goal is to find a minimum cardinality set of cliques that cover all the edges of a graph---is a classic NP-hard problem that has received much attention from both the theoretical and experimental algorithms communities. While small sparse graphs can be solved exactly via the branch-and-reduce algorithm of Gramm et al. [JEA 2009], larger instances can currently only be solved inexactly using heuristics with unknown overall solution quality.
We revisit computing minimum ECCs exactly in practice by combining data reduction for both the ECC \emph{and} vertex clique cover (VCC) problems. We do so by modifying the polynomial-time reduction of Kou et al. [Commun. ACM 1978] to transform a reduced ECC instance to a VCC instance; alternatively, we show it is possible to ``lift'' some VCC reductions to the ECC problem. 
Our experiments show that combining data reduction for both problems (which we call \emph{synergistic data reduction}) enables finding exact minimum ECCs orders of magnitude faster than the technique of Gramm et al., and allows solving large sparse graphs on up to millions of vertices and edges that have never before been solved. With these new exact solutions, we evaluate the quality of recent heuristic algorithms on large instances for the first time. The most recent of these, \textsf{EO-ECC} by Abdullah et al. [ICCS 2022], solves 8 of the 27 instances for which we have exact solutions. It is our hope that our strategy rallies researchers to seek improved algorithms for the ECC problem.
\end{abstract}

\clearpage

%\vfill 
%\ \pagebreak \ \vfill 
%\ \pagebreak
% \pagestyle{plain}

%\pagenumbering{arabic}

%%%%%%%%%%%%%%%%%%%%%%%%%%%%%%%%%%%%%%%%%%%%%%%%%%%%%%%%%%%%%%%%%%%%%%%%%%%%%%%%%%%%%%%%%%%%%%%%%%%%%%%%%%%%
\section{Introduction}
In the \emph{edge clique cover (ECC) problem}, also called the \emph{clique cover} problem, we are given an unweighted, undirected, simple graph $G=(V,E)$ and asked to find a minimum cardinality set of cliques that cover the edges of $G$. The ECC problem is NP-hard, however its decision variant did not appear in Karp's original list of NP-complete problems~\cite{karp-1972}, though the \emph{vertex clique cover (VCC) problem} did. Compared to the VCC problem, the ECC problem has received the lion's share of attention from researchers, in part because it has many applications. For instance, edge clique covers can be used to succinctly represent constraints for integer program solvers~\cite{atamturk-2000} and to detect communities in networks~\cite{conte-2020}. 
%It has further applications through the complementary graph coloring problem --- as finding a minimum vertex clique cover in $G$ is equivalent to finding a proper vertex coloring with minimum number of colors in $\overline{G}$ --- including the canonical register allocation problem~\cite{chaitin-1982} and its prolific use in scheduling.%~\cite{}. 

Data reduction rules, which allow one to transform an input instance to a smaller equivalent instance of the same problem, are powerful tools for solving NP-hard problems in practice~\cite{akiba-tcs-2016,lamm2017finding}. Of particular interest in the field of parameterized algorithms is whether the repeated application of data reduction rules produces a \emph{kernel}---which is a problem instance that has size bounded by a function $O(f(k))$ of some parameter $k$ of the input. Gramm et al.~\cite{gramm-2009} show that repeated application of four simple reduction rules produce a kernel of size $2^k$, where the parameter $k$ is the number of cliques in the cover.  When intermixed with branch-and-bound (a so-called \emph{branch-and-reduce} algorithm), these reduction rules enable solving sparse graphs of up to 10,000 vertices quickly in practice. Since their seminal work, no progress has been made on solving larger instances exactly. Indeed, the prospect of doing so is grim since polynomial kernels are unlikely to exist for the ECC problem, when parameterized on the solution size~\cite{cygan-2014}. Although researchers have found further FPT algorithms (and smaller kernels) with other parameters~\cite{blanchette-2012,ullah-2022}, these algorithms are still only able to solve relatively small instances in practice. The outlook for the VCC problem is even worse in theory: it is unlikely to have any problem kernel when parameterized on the number of cliques $k$ in the cover, as it is already NP-hard for $k=3$ (since it is equivalent to 3-coloring the complement graph).

However, recent data reductions for the \emph{VCC} problem have been shown to significantly accelerate computing minimum VCCs exactly in practice. Strash and Thompson~\cite{strash-2022} introduce a suite of reduction rules and show that data reduction can solve real-world sparse graphs with up to millions of vertices in seconds. 

\paragraph*{Our Results} We show that combining VCC and ECC data reductions enables the ECC problem to be solved exactly on large instances not previously solvable by Gramm et al.~\cite{gramm-2009}. We do so by modifying the polynomial-time transformation of Kou et al.~\cite{10.1145/359340.359346} to transform a reduced ECC instance to a VCC instance, but also show that some VCC reductions can be ``lifted'' to ECC reductions. Their combined reduction power (which we call \emph{synergistic data reduction}) reduces an ECC instance significantly more than Gramm et al.'s reductions alone, enabling us to exactly solve graphs with millions of vertices and edges. With these exact results, we objectively evaluate the quality of heuristic algorithms recently introduced in the literature. On instances not solvable exactly with our method, we give upper and lower bounds for use by future researchers. 

\section{Related Work}
\label{sec:related-work}

 We now briefly review the relevant previous work on the ECC and VCC problems, as well as practical data reduction in related problems.
 \subsection{Edge Clique Cover}

The goal of the edge clique cover (ECC) problem is to cover the edges of the graph $G$ with a minimum number of cliques, denoted $\theta_E(G)$. That is, to find a set of cliques $\mathcal{C} = \{C_1, C_2,\ldots, C_k\}$ such that each edge is in at least one clique in $\mathcal{C}$ and $k=\theta_E(G)$. Although closely related to the VCC problem (to cover \emph{vertices} with a minimum number of cliques, denoted $\theta(G)$), Brigham and Dutton~\cite{brigham-1983} showed that $\theta(G) \leq\theta_E(G)$, and that these cover numbers can differ significantly: $\theta_E(G)$ can be as large as $\theta(G)(n - \theta(G))$. Gramm et al.~\cite{gramm-2009} introduced four data reductions for the ECC problem, which they show can solve real-world sparse graphs of hundreds of vertices, as well as synthetic instances on up to 10K vertices in practice, when interleaved with branch and bound. Furthermore, they showed that their data reductions produce a kernel of size $2^k$, where $k$ is the number of cliques. Cygan et al.~\cite{cygan-2014} showed that it is unlikely that a polynomial-size kernel exists when parameterized by the number of cliques in the cover, as otherwise the polynomial hierarchy collapses to its third level.  
%In contrast, for the highly-related (edge) clique partition problem, where we wish to partition the edges of the graph $G$ into a minimum number of cliques, Mujuni and Rosamond~\cite{mujuni2008para} showed that this problem has a size $k^2$ kernel, where $k$ is the number of cliques.
However, Blanchette et al.~\cite{blanchette-2012} gave a linear-time algorithm having running time $O(2^{\binom{k}{2}}n)$ where $k$ is the treewidth of the graph. In practice, their algorithm is effective on graphs with hundreds of vertices and small treewidth.
For larger graphs, heuristic methods are used to compute inexact ECCs~\cite{conte-2020,abdullah-2021,abdullah-2022} in practice. No heuristic algorithm performs best on all instances, and their overall quality is unclear.

\subsection{Vertex Clique Cover}
The vertex clique cover (VCC) problem is NP-hard, and closely related to the maximum independent set and graph coloring problems. The size of a minimum VCC (also called the clique cover number) $\theta(G)$ is lower bounded by the size of a maximum independent set (the independence number $\alpha(G)$) and equivalent to the chromatic number of the complement graph, $\chi(\overline{G})$. There is a rich line of research on the graph coloring problem, which seeks to compute the chromatic number; many of the theoretical results for the VCC problem come via the graph coloring problem. The fastest exact exponential-space algorithm for computing the chromatic number on an $n$-vertex graph has time $O^*(2^n)$ (where $O^*$ hides polynomial factors) using a generalization of the exclusion-inclusion principle~\cite{koivisto2006an}, and in polynomial space the problem can be solved in time $O(2.2356^n)$~\cite{gaspers2017faster}. Furthermore, there exists no polynomial-time algorithm with approximation ratio better than $n^{1-\epsilon}$ for $\epsilon>0$ unless $P=NP$~\cite{zuckerman2007linear}.

In terms of data reduction, we note that it is unlikely that a kernel exists when parameterized on the (vertex) clique cover number. Deciding if a cover with even $3$ cliques exists is NP-complete (since $3$-coloring the complement is NP-hard). A polynomial kernel would have size $O(1)$ and could be computed in polynomial time. Solving the kernel with brute-force computation would solve the VCC problem in polynomial time, implying $P=NP$. However, in practice, the VCC problem can be solved on large, sparse real-world graphs using the data reductions by Strash and Thompson~\cite{strash-2022}.

\subsection{Data Reduction in Practice for Related Problems}
\label{sec:other}
Other classical NP-hard problems have large suites of data reductions that are effective in practice, including minimum vertex cover~\cite{akiba-tcs-2016,fellows2018known}, maximum cut~\cite{ferizovic-2020}, and cluster editing~\cite{blasius-2022}. Popular data reductions include variations of simplicial vertex removal, degree-2 folding, twin, domination, unconfined, packing, crown, and linear-programming-relaxation-based reductions~\cite{akiba-tcs-2016}. Even the simplest reductions can be highly effective when combined with other techniques~\cite{chang2017,strash-power-2016}. Data reductions are most effective in sparse graphs, which are the graphs that we consider here. Finally, similar to what we propose here, other NP-hard problems are solved by first applying a problem transformation. In particular, algorithms for minimum dominating set problem first transform the problem to an instance of the set cover problem~\cite{vanrooij-2011}.

\section{Preliminaries}

We consider a simple finite undirected graph $G = (V, E)$ with vertex set $V$ and edge set $E \subseteq \{\{u,v\}\mid u,v\in V\}$. For brevity, we denote by $n=|V|$ and $m=|E|$ the number of vertices and edges in the graph, respectively. When more specificity is needed, we denote the vertex and edge set of a graph $G$ by $V(G)$ and $E(G)$ respectively. We say two vertices $u,v \in V$ are \emph{adjacent} (or \emph{neighbors}) when $\{u,v\} \in E$. The \emph{open neighborhood} of a vertex $v \in V$ is the set of its neighbors $N(v) := \{u \mid \{u, v\} \in E\}$, and the \emph{degree} of $v$ is $|N(v)|$. We further define the \emph{closed neighborhood} of a vertex $v \in V$ to be $N[v] := N(v) \cup \{v\}$. Extending these definitions, the open neighborhood of a set $A \subseteq V$ is $N(A):=\bigcup_{v \in A} N(v)\setminus A$ and the closed neighborhood of $A$ is $N[A] := \bigcup_{v \in A} N[v]$. The subgraph of $G$ induced by a vertex set $V'\subseteq V$, denoted $G[V']$, has vertex set $V'$ and edge set $E'=\{\{u,v\}\in E \mid u,v\in V'\}$.
The \emph{degeneracy} $d$ of a graph $G$ is the smallest value such that every nonempty subgraph of $G$ has a vertex of degree at most $d$~\cite{lick_white_1970}. It is possible to order the vertices of a graph $G$ in time $O(n+m)$ so that every vertex has $d$ or fewer neighbors later in the ordering; such an ordering is called a \emph{degeneracy} ordering~\cite{els-2013}.

A vertex set $C \subseteq V$ is called a \emph{clique} if, for each pair of distinct of vertices $u, v \in C$, $\{u,v\} \in E$.
A set of cliques $\mathcal{C}$ is called an \emph{edge clique cover (ECC)} (or just a \emph{clique cover}) of $G$ if for every edge $\{u,v\}\in E$ there exists at least one $C\in\mathcal{C}$ such that $\{u,v\}\subseteq C$. That is, there is some clique in $\mathcal{C}$ that \emph{covers} $\{u,v\}$. The set of cliques $\mathcal{C}$ is said to cover the graph $G$. An ECC of minimum cardinality is called a \emph{minimum ECC}, and its cardinality is denoted by $\theta_E(G)$, called the edge clique cover number. 
%A clique cover $\mathcal{C}$ is (inclusion) \emph{minimal} if there is no proper subset $X \subset \mathcal{C}$ that is a clique cover of $G$.  Note that a minimum ECC is a minimal cover, but not all minimal covers have size $\theta_E(G)$.

Similarly, in a \emph{vertex clique cover (VCC)}, every \emph{vertex} $v\in V$ is covered by some clique. The cardinality of a minimum VCC is the \emph{clique cover number}, denoted by $\theta(G)$.

\section{Existing Tools Discussion}
In this section, we discuss basic tools that we will use to solve the ECC problem, together with insights into their behavior on sparse graphs. We begin by describing the existing ECC data reductions by Gramm et al.~\cite{gramm-2009}. We then discuss how to convert an input ECC instance to an equivalent VCC instance using the technique of Kou et al.~\cite{10.1145/359340.359346}. We will extend these tools to develop our full algorithm combining ECC and VCC reductions in the next section.
%
%Our strategy is to employ three steps to preprocess the input ECC instance $G$ into a new (and much smaller) instance of the VCC problem: (1) apply data reduction rules for the ECC problem, producing a smaller instance of the ECC problem $G'$ (2) we then perform a \emph{problem} reduction from the ECC problem to the VCC problem, producing a new instance of the VCC Problem $G_{VCC}$, and (3) we then apply data reduction rules for the VCC problem, producing a new graph smaller instance of the VCC problem.
%
%Once a problem is converted to a VCC instance, any VCC (or graph coloring) solver can then be used to solve the instance either exactly or heuristically. We further discuss how to ``lift'' VCC reductions to ECC instances, which allows data reduction without conversion.

\subsection{ECC Reduction Rules}
Gramm et al.~\cite{gramm-2009} introduce four data reduction rules that either cover edges by a clique known to be in a minimum cardinality ECC or add edges to the input graph $G$. Once all of a vertex $v$'s incident edges are covered, $v$ can be removed from the graph. 

With each edge $\{u,v\}$, Gramm et al. store the common neighbors in $G$, denoted by $N_{\{u,v\}}$, as well as a count $c_{\{u,v\}} = |E(G[N_{\{u,v\}}])|$ of the edges between common neighbors. These values are updated in ECC Reduction~\ref{red:gramm_1}, and are used in ECC Reduction~\ref{red:gramm_2}. 

Throughout the application of data reductions, vertices are removed from $G$ and edges are covered. Figure~\ref{fig:gramm} illustrates an example of the reductions.  Set let edge set $E'\subseteq E$ be the set of uncovered edges (by extension, $E\setminus E'$ are the covered edges). The graph $G$ only changes when a vertex is removed.
%We let $G[E']$ be the subgraph of $G$ induced by the (uncovered) edge set $E'$, which has vertex set $V(G[E']) = \cup_{\{u,v\}\in E'}\{u,v\}$ and edge set $E(G[E']) = \{\{u,v\}\in E\mid \{u,v\}\subseteq V'\}$. After all data reductions are applied, $G = G[E']$.

We note that the data reductions by Gramm et al.~\cite{gramm-2009} are particularly effective for sparse graphs; however, the original data reductions were not written with efficiency in mind. Although these reductions have (very) slow theoretical running times, we offer insights as to why their reductions are faster in practice than indicated by the theoretical running time from Gramm et al.~\cite{gramm-2009}.

\begin{figure}
\begin{center}
\includegraphics{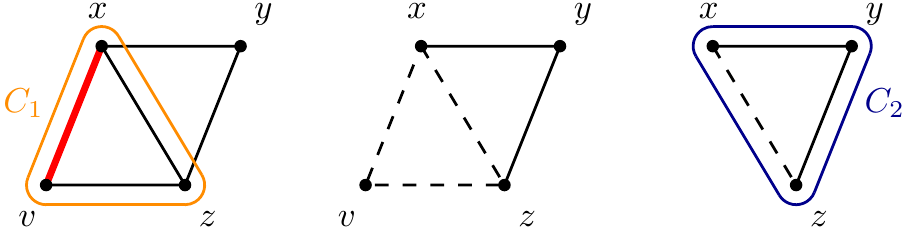}
\caption{Illustrating Gramm et al.~\cite{gramm-2009}'s data reductions: (left) edge $\{v,x\}$ is in exactly one maximal clique $C_1$, triggering ECC Reduction~\ref{red:gramm_2} and covering edges $\{v,x\}$, $\{v,z\}$, and $\{x,z\}$ (middle). Vertex $v$ can then be removed with ECC Reduction~\ref{red:gramm_1}. The remaining triangle is covered by clique $C_2$ by applying ECC Reduction~\ref{red:gramm_2} to either $\{x,y\}$ or $\{y,z\}$.}
\label{fig:gramm}
\end{center}
\end{figure}

\begin{eccreduction}[\cite{gramm-2009}]
Let $v\in V$ be a vertex whose incident edges are all covered (i.e., in $E\setminus E'$). Then remove $v$ from the graph $G$, along with its incident edges, and update values $c_{\{w,x\}}$ and $N_{\{w,x\}}$ for all uncovered edges $\{w,x\} \in E'$ whose endpoints are both adjacent to $v$, i.e., $\{w,x\} \subseteq N(v)$.
\label{red:gramm_1}
\end{eccreduction}

As noted by Gramm et al.~\cite{gramm-2009}, this step can be applied to all vertices in running time $O(n^2m)$ by iterating over each vertex $v$ and updating $N_{\{u,w\}}$ for all edges $\{u,w\}\in E'$ whose endpoints are adjacent to $v$. However, in sparse graphs the maximum degree in $G$, denoted $\Delta$, is significantly smaller than $n$. Each edge $\{u,w\}$ has its set $N_{\{u,w\}}$ updated at most $\Delta$ times, taking $O(\Delta)$ time to update each time, giving a more reasonable running time of $O(\Delta^2 m)$. We note that with adjustments, this can be run faster by enumerating all triangles in $G$ in time $O(dm)$ using the triangle listing by Chiba and Nishizeki~\cite{chiba-1985} and updating $N_{\{u,w\}}$ for edge $\{u,w\}$ in each triangle; however, this is a different implementation than that done by Gramm et al.~\cite{gramm-2009} and not our focus here.

\begin{eccreduction}[\cite{gramm-2009}]
\label{red:gramm_2}
Let edge $\{u,v\}\in E'$ be an uncovered edge such that $c_{\{u,v\}} = \binom{|N_{\{u,v\}}|}{2}$ (i.e., the edge is in exactly one maximal clique in $G'$). Then $C=N_{\{u,v\}}\cup\{u,v\}$ is a maximal clique of $G$ in some minimum ECC. Add the clique $C$ to the clique cover, and cover any uncovered edges in $C$ in $G$. 
\end{eccreduction}

As noted by Gramm et al.~\cite{gramm-2009}, ECC Reduction~\ref{red:gramm_2} can be implemented in time $O(n^2m)$ by iterating over each edge $\{u,v\}\in E'$, checking if $c_{\{u,v\}} = \binom{|N_{\{u,v\}}|}{2}$ in $O(1)$ time, and covering the edges of $\{u,v\}$'s clique in time $O(n^2)$ time.

However, when run on sparse graphs, which tend to have low degeneracy $d$~\cite{els-2013}, this rule is much faster. Graphs with degeneracy $d$ have cliques of at most $d+1$ vertices, therefore the reduction is only triggered when $|N_{\{u,v\}}| < d$. Hence, in practice, we should observe the much faster running time of $O(d^2m)$.

Gramm et al. introduce two more ECC reductions, however, they are more complex and we choose not to run them here. Experiments by Gramm et al. show that these reductions are very slow in practice, and only improve the search tree size by a constant factor when incorporated in branch and reduce~\cite{gramm-2009}. We invite the interested reader to see ECC Reductions~\ref{red:gramm_3} and~\ref{red:gramm_4} in Appendix~\ref{appendix:gramm}.

\subsection{Transforming an ECC Instance to a VCC Instance}
Kou et al.~\cite{10.1145/359340.359346} showed that the ECC problem is NP-hard via a polynomial-time reduction from the VCC problem. Furthermore, they gave a polynomial-time reduction \emph{to} the VCC problem, which we use as the basis of our transformation. We describe their transformation and briefly justify why it works.

\begin{figure}[!t]
\begin{subfigure}[t]{\textwidth}
\centering
\includegraphics[]{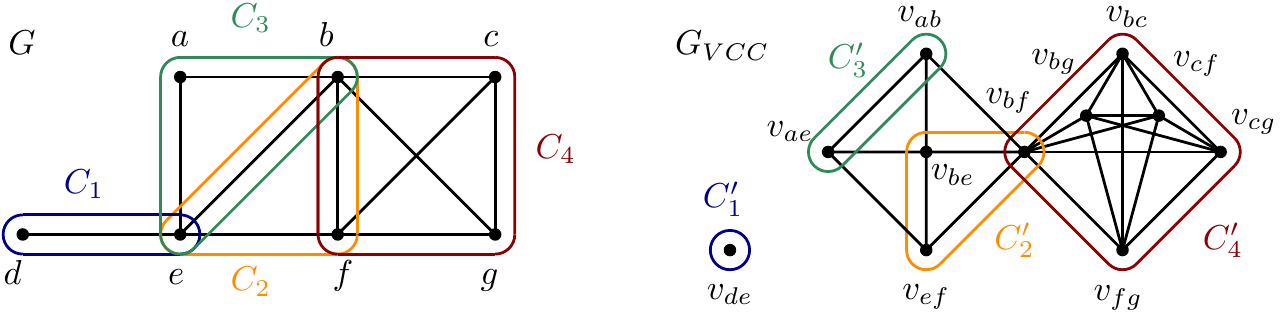}
\end{subfigure}
\caption{Example graph $G$ and its transformed graph $G_{VCC}$, with minimum clique covers.}
\label{fig:kou}
\end{figure}

Given an input graph $G=(V,E)$ for the ECC problem Kou et al.~\cite{10.1145/359340.359346} transform $G$ to a new graph $G_{VCC} = (V_{VCC}, E_{VCC})$ that is an equivalent VCC instance as follows. For each edge $\{x,y\}\in E$, create a new vertex $v_{xy}\in V_{VCC}$, then add an edge $\{v_{xy}, v_{wz}\}$ to $E_{VCC}$ if and only if there exists a clique $C$ in $G$ containing both $\{x,y\}$ and $\{w,z\}$. Now, for any given subset $C\subset V_{VCC}$, $C$ is a clique in $G_{VCC}$ iff its vertices' corresponding edges in $E$ also induce a clique in $G$. Hence, a minimum cardinality VCC in $G_{VCC}$ corresponds to a minimum cardinality ECC in $G$. (See Figure~\ref{fig:kou}.)

To determine if two edges are in a clique together in $G$, Kou et al.~\cite{10.1145/359340.359346} make the following observation: 

\begin{observation}[\cite{10.1145/359340.359346}]
\label{obs:triquad}
Two distinct edges $\{x,y\}, \{w,z\}$ are in a clique together in $G$ iff $\{x,y\}$ and $\{w,z\}$ are incident and $\{x,y\} \cup \{w,z\}$ induce a triangle, or $\{x,y\}$ and $\{w,z\}$ are not incident and $\{w,x,y,z\}$ form a 4-clique.
\end{observation}

However, there is a clear issue when using this transformation: how large can $G_{VCC}$ be? We briefly discuss its size and sparsity.

\subsubsection{The Effect of Transformation on Graph Size and Sparsity}
In the worst case, the size of $G_{VCC}$ is a quadratic factor larger than $G$. Indeed, if the graph $G$ is itself the complete graph $K_n$, on $n$ vertices and $\Theta(n^2)$ edges, then the transformed graph is the complete graph $K_{n(n-1)/2}$ having $\Theta(n^2)$ nodes and $\Theta(n^4)$ edges. However, we show that the size of the graph only increases by a factor of $O(d^2)$, where $d$ is the degeneracy of the graph. Real-world sparse graphs have low degeneracy~\cite{els-2013}, and thus this is a significant improvement over the worst case.

\begin{theorem}
Let the degeneracy of $G=(V,E)$ be $d$. Then $|V_{VCC}| = m \leq dn$ and $|E_{VCC}| = O(d^2m)$.
\end{theorem}

\begin{proof}
By construction $|V_{VCC}| = m$; hence, to bound $|V_{VCC}|$, we bound the number of edges in $G$. In a degeneracy ordering of the graph, each vertex has at most $d$ later neighbors in the ordering. Therefore, $|V_{VCC}| = m \leq dn$. 
To bound $|E_{VCC}|$, we compute an upper bound on the number of triangles and 4-cliques in $G$. Following Observation~\ref{obs:triquad}, each edge in $E_{VCC}$ corresponds to a pair of edges in $E$ contained in a triangle or a pair of non-incident edges in a 4-clique. Each triangle has 3 edges, and each 4-clique has 3 pairs of non-incident edges. Therefore, an asymptotic upper bound of the number of triangles and 4-cliques in $G$ gives an upper bound for $|E_{VCC}|$.  
In any triangle, some vertex must come first in a degeneracy ordering, and can be in a triangle with at most $\binom{d}{2}$ of its at most $d$ later neighbors. Therefore each vertex is in $O(d^2)$ triangles with its later neighbors and, summing up over all vertices, contributes at most $O(d^2n)$ edges to $E_{VCC}$. Similarly, for each edge $\{u,v\}$ we count the number of $4$-cliques it is in with (non-incident) edges that come lexicographically after it in the degeneracy ordering. The number of triangles the second vertex can be in with later neighbors is $\binom{d}{2}$ and hence the edge is in at most $O(d^2)$ 4-cliques with $v$'s at most $d$ later neighbors, giving at most $O(d^2m)$ 4-cliques total. Thus, we conclude that $|E_{VCC}| = O(d^2m)$.
\end{proof}

Thus, the size of the $G_{VCC}$ has size at most $O(d^2m)$, a factor $O(d^2)$ larger than $G$.
As a consequence, the average degree of the graph may increase, but by no more than a factor $d$: whereas $G$ has average degree $2|E|/|V| = O(dn) / n = O(d)$, graph $G_{VCC}$ has average degree $2|E_{VCC}|/|V_{VCC}| = O(d^2m) / m = O(d^2)$. Therefore, for input graphs with small degeneracy, the transformed graph is expected to be sparse as well.

However, even if the degeneracy $d$ is small, the graph $G_{VCC}$ may be very large in practice. Hence, to use this transformation, we require techniques to keep the graph size manageable.

\section{Synergistic Reductions: Applying ECC and VCC Reductions}

We propose to handle the blow-up by Kou et al.~\cite{10.1145/359340.359346} by applying both ECC and VCC reductions to the problem, which we call \emph{synergistic} data reduction. We first show how to adjust the transformation to work on reduced ECC instances, after which we can apply VCC reductions. We also explore the possibility of ``lifting'' VCC reductions to ECC reductions.

\subsection{Transforming a Partially-Covered ECC Problem Kernel}
Recall that the data reductions from Gramm et al.~\cite{gramm-2009} result in a graph in which some edges are covered, which is not supported by the transformation of Kou et al.~\cite{10.1145/359340.359346}. While it is tempting to modify the transformation to operate on \emph{only} the uncovered edges $E'$, this does not necessarily result in an equivalent instance, as already-covered edges may still be needed to compute a minimum number of cliques covering $E'$. For instance, in Figure~\ref{fig:gramm}, covering edges $\{x,y\}$ and $\{y,z\}$ with the single clique $C_2$ uses the already-covered edge $\{x,z\}$.

One way to correct for this is to first perform the transformation on the entire graph $G = (V,E)$, and then take the subgraph induced by the vertices corresponding to uncovered edges in $E'$. However, this strategy is slow when the edge set $E$ is significantly larger than $E'$.
We show that it is possible to perform the transformation without making vertices for all edges in $E$. Note that since all that remains is to cover the edges in $E'$, we now focus on covering all $E'$ using a minimum number of cliques in $G$. Taken together with already-chosen cliques from ECC reductions, this gives us a covering of all of $G$. (See Figure~\ref{fig:transformkou}.)

\begin{figure}[!t]
\begin{subfigure}[t]{\textwidth}
\centering
\includegraphics[]{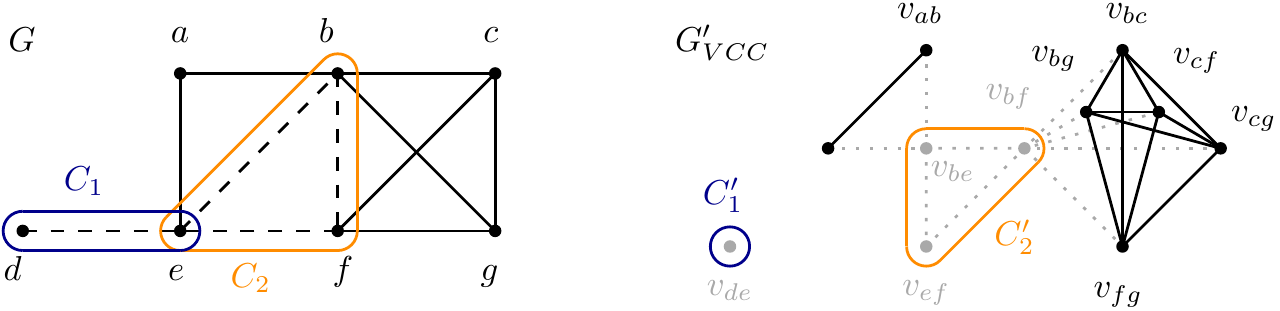}
\end{subfigure}
\caption{A partially-covered graph $G$ with cliques $C_1$, $C_2$ already added to the cover, and its transformed graph $G'_{VCC}$. Grayed vertices and (dotted) edges are those in $G_{VCC}$, but not $G'_{VCC}$.}
\label{fig:transformkou}
\end{figure}

We transform $G$ to a graph $G'_{VCC} = (V'_{VCC}, E'_{VCC})$, where $V'_{VCC} = \{v_{xy}\mid \{x,y\}\in E'\}$ and $E'_{VCC} = \{\{v_{xy}, v_{wz}\}\mid \{x,y\},\{w,z\}\in E' \text{ and $\{x,y\}\cup\{w,z\}$ is a clique in $G$}\}$. This transformation preserves cliques in $G$ that cover edges in $E'$, which we capture with the following observation.

\begin{observation}
\label{obs:same}
If $C'$ is a clique in $G'_{VCC}$ then $C=\cup_{v_{xy}\in C'}\{x,y\}$ is a clique covering $|C'|$ edges of $E'$ in $G$.
\end{observation}

Furthermore, the transformation gives a correspondence between cliques covering $E'$ in $G$ and VCCs in $G'_{VCC}$.

\begin{theorem}
\label{thm:equivalent}
If $\mathcal{C}'$ is a VCC in $G'_{VCC}$ then $\mathcal{C} = \{\cup_{v_{xy}\in C'}\{x,y\}\mid C'\in \mathcal{C'}\}$ is a set of cliques covering $E'$ in $G$.
\end{theorem}
\begin{proof}
By Observation~\ref{obs:same}, every clique $C'\in\mathcal{C}'$ in $G'_{VCC}$ corresponds to a clique $C = \cup_{v_{xy}\in C'}\{x,y\}$ in $G$ that covers its corresponding edges of $E'$. Hence, a VCC that covers all $V'_{VCC}$ of $G'_{VCC}$ corresponds to a collection of cliques covering all edges $E'$ in $G$.
\end{proof}

Note that in Theorem~\ref{thm:equivalent}, $|\mathcal{C}|=|\mathcal{C}'|$. Hence, a minimum VCC in $G'_{VCC}$ corresponds to a minimum-cardinality set of cliques covering $E'$ in $G$. This transformation gives us a technique for computing a minimum ECC: First apply the data reductions of Gramm et al., then compute $G'_{VCC}$ and use VCC reductions combined with any VCC solver to compute a minimum VCC in $G'_{VCC}$, giving us cliques covering $E'$ in $G$ and, ultimately an entire ECC of $G$. While applying VCC reductions to $G'_{VCC}$ may produce a smaller instance, these data reductions are not actually producing a smaller \emph{ECC instance}. However, as we now show, we can also ``lift'' some VCC reductions to the ECC problem, by keeping the equivalence between cliques in the transformation in mind. 

\subsection{Lifting VCC Reduction Rules to ECC}
Unlike the ECC problem, the VCC problem has many data reduction rules~\cite{strash-2022}. These include reductions based on simplicial vertices, dominance, twins, degree-2 folding, and crowns. We briefly discuss two classes of VCC reductions: clique-removal-based rules and folding-based rules. We place them in the context of the ECC problem, and discuss whether it is viable to ``lift'' them to the ECC problem, and if the graph transformation is needed.
%As we apply these to our transformed instance, we briefly discuss each data reduction. For proofs of correctness, see the work by Strash and Thompson~\cite{strash-2022}.
%
By combining existing ECC reductions with VCC reductions, we aim to reduce ECC instances even further. 

\subsubsection{Clique-Removal-Based VCC Reductions}
We call a VCC reduction that removes a set of cliques from the graph a \emph{clique-removal-based} rule.
Four VCC reductions (simplicial vertex, dominance, twin removal, and crown) are clique-removal-based rules~\cite{strash-2022}.
Such rules can be easily transformed into an ECC reduction: By the equivalence between cliques in the problem transformation, stated in Observation~\ref{obs:same}, removing a clique in $G'_{VCC}$ is equivalent to covering its corresponding clique in $G$. Thus, to apply clique-removal-based VCC reductions directly to the ECC problem, we can compute $G'_{VCC}$, apply any clique-removal-based rules, and then cover these cliques in $G$. We capture this with the following theorem.

\begin{theorem} Any clique-removal-based VCC reduction can be lifted to an ECC reduction.
\end{theorem}

Of course, we could try to apply these reductions more efficiently to $G$ directly. We discuss two clique-removal-based VCC reductions and discuss whether they are worth implementing for ECC directly, or if we should transform $G$ to $G'_{VCC}$ first.

\paragraph*{Simplicial Vertex Reduction}

\begin{figure}[!t]
\begin{subfigure}[t]{0.27\textwidth}
\centering
\includegraphics[]{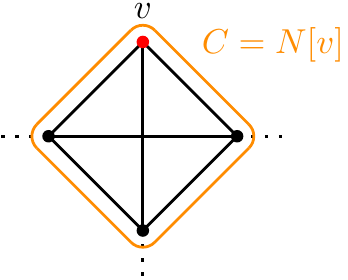}
\caption{$C$ is in a minimum VCC.}
\label{fig:simplicial1}
\end{subfigure}\hfill
\begin{subfigure}[t]{0.67\textwidth}
\centering
\includegraphics[]{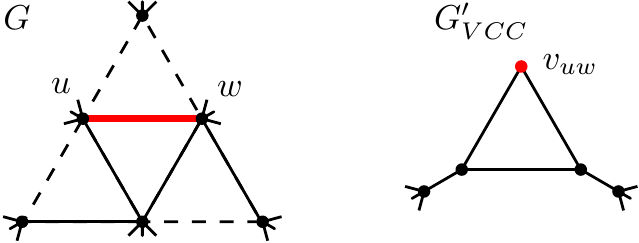}
\caption{Edge $\{u,w\}\in E'$ is in two cliques in $G$, but $v_{uw}$ in $G'_{VCC}$ is simplicial.}
\label{fig:simplicial2}
\end{subfigure}
\caption{The simplicial vertex VCC reduction can be applied after transforming $G$ to $G'_{VCC}$.}
\end{figure}

A vertex $v$ is \emph{simplicial} if $N[v]$ forms a clique. In this case, the clique $C = N[v]$ is in some minimum VCC. (See Figure~\ref{fig:simplicial1}.)

\begin{vccreduction}[Simplicial Vertex Reduction~\cite{strash-2022}]
\label{vcc:simplicial}
Let $v \in V$ be a simplicial vertex. Then $C=N[v]$ is a clique in some minimum VCC. Add $C$ to the clique cover and remove $C$ from the graph.
\end{vccreduction}

Applying VCC Reduction~\ref{vcc:simplicial} on $G'_{VCC}$ is reminiscent of applying ECC Reduction~\ref{red:gramm_2} on the untransformed graph $G$.  While it is true that for a $\{u,w\}\in E'$, if $N_{\{u,w\}}$ is a clique in $G$, then $v_{uw}$ is simplicial in $G'_{VCC}$, the converse is not true in general. Hence, VCC Reduction~\ref{vcc:simplicial} is more powerful. Consider the counterexample in Figure~\ref{fig:simplicial2}. Vertex $v_{uw}$ is simplicial in $G'_{VCC}$, but $\{u,w\}\in E'$ is in two cliques of $G$.

%
%TODO: Figure
%Consider the partially covered graph and its transformation in Figure~\ref{fig:simplicial}. VCC Reduction~\ref{vcc:simplicial} applies in $G_{VCC}$, but ECC Reduction~\ref{red:gramm_2} does not apply in $G[E']$.
%
Thus, we have new data reduction for the ECC problem, which subsumes ECC Reduction~\ref{red:gramm_2}:

\setcounter{eccreduction}{4}
\begin{eccreduction}[Lifted Simplicial Vertex Reduction]
\label{ecc:simplicial}
Let edge $\{u,w\}\in E'$ and let set $C=\{x,y\in V \mid \{x,y\}\in E' \text{ and } \{u,w\}\cup\{x,y\} \text{ is a clique in $G$}\}$ be the set of vertices of edges in some clique with $\{u,w\}$. If $C$ is a clique, then add $C$ to the clique cover, and cover any uncovered edges of $C$ in $G$. 
\end{eccreduction}
%\begin{proof}
%Vertex $v_{\{u,w\}}$ in $G_{VCC}$ is simplicial and hence can be removed by VCC Reduction~\ref{vcc:simplicial}. This is equivalent to covering the edges $\{x,y\}\in E'$ for each $v_{\{x,y\}}\in N[v_{\{u,w\}}]$ by a single clique.
%\end{proof}

To apply our lifted reduction, we could of course first compute $G'_{VCC}$ and then apply VCC Reduction~\ref{vcc:simplicial}. However, we can also apply it directly to $G$ with a slight modification to ECC Reduction~\ref{red:gramm_2}. For each edge $\{u,w\}\in E'$ compute the common neighborhood $N_{\{u,w\}}$. Instead of checking that the common neighborhood is a clique, collect the uncovered edges between vertices in $N_{\{u,w\}}$, and check if they induce a clique. Since $|N_{\{u,w\}}|\leq \Delta$, it takes $O(d\Delta)$ to collect uncovered edges by iterating through the at most $d$ later neighbors of each vertex, which dominates the running time of this step. Exhaustively applying the reduction to all edges takes time $O(d\Delta m)$, which is slightly slower than the $O(d^2m)$ time for ECC Reduction~\ref{red:gramm_2}.

Is it worth applying ECC Reduction~\ref{ecc:simplicial} directly to $G$, or should we first transform $G$ and run VCC Reduction~\ref{vcc:simplicial} instead? The transformation can be done in time $O(d^2m)$ by enumerating all of the triangles and $4$-cliques of $G$~\cite{chiba-1985}, hence performing the transformation is faster in theory than applying ECC Reduction~\ref{ecc:simplicial} to $G$ directly. However, in $G'_{VCC}$ the largest clique may have as many as $\Theta(d^2)$ vertices and $\Theta(d^4)$ edges since a clique of size $d+1$ in $G$ has $\Theta(d^2)$ edges in $G$. Therefore, the time to apply the VCC Reduction~\ref{vcc:simplicial} for each of the $m$ vertices of $G'_{VCC}$ is $O(d^4m)$. Thus, in theory, it is more efficient to apply ECC Reduction~\ref{ecc:simplicial} directly, rather than first applying a conversion. 

However, there are compelling reasons to perform the conversion. For one, most implementations of simplicial vertex reductions limit the degree of the vertex considered -- in some cases to as small as two -- since large-degree simplicial vertices rarely appear in sparse graphs. Therefore, in practice, it is unlikely that we would observe this large running time. However, a more compelling reason to perform the transformation is that there are two highly effective VCC reductions that we do not know how to apply directly to $G$. The first is the crown removal reduction (a clique-removal-based reduction) and the second is the degree-2 folding-based reduction.

\paragraph*{Crown Removal Reduction}

\label{appendix:crown}
The crown removal reduction is arguably one of the most powerful data reductions, successfully reducing sparse instances for the minimum vertex cover and VCC problems~\cite{abu-khzam-2007,akiba-tcs-2016,chang2017}.

\begin{figure}[!t]
\begin{subfigure}[t]{0.33\textwidth}
\centering
\includegraphics[]{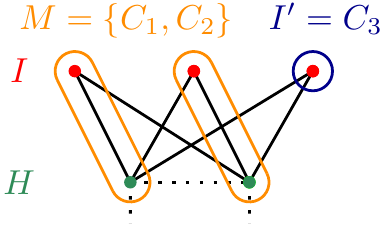}
\caption{A flared crown. Cliques $C_1$, $C_2$, and $C_3$ are in some minimum VCC.}
\label{fig:crown1}
\end{subfigure}\hfill
\begin{subfigure}[t]{0.60\textwidth}
\centering
\includegraphics[]{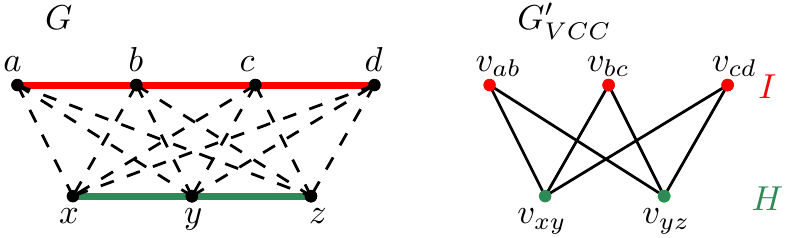}
\caption{A partially-covered $G$ where $G'_{VCC}$ is a flared crown.}
\label{fig:crown2}
\end{subfigure}
\caption{The crown removal VCC reduction can be applied after transforming $G$ to $G'_{VCC}$.}
\end{figure}

In a pair of vertex sets $(H,I)$, $H$ is called a \emph{head} and $I$ a \emph{crown} if: $I$ is an independent set, $N(I) = H$, and there exists a matching from $H$ to $I$ of size $|H|$. Figure~\ref{fig:crown1} shows a crown structure. Note that, due to the matching requirement, $|I| \geq |H|$. If $|I|=|H|$, the crown is called \emph{straight}, otherwise it is \emph{flared}. 
%As we now show, each edge in the matching, and each unmatched vertex in $I$, is in some minimum vertex clique cover, thus the vertices $H\cup I$ can be removed from the graph, and the $\theta(G') = \theta(G') - |H|$. number of the remaining graph is  
Strash and Thompson~\cite{strash-2022} give the following data reduction for the VCC problem, adapting a data reduction for the dual coloring problem~\cite{fomin-2019}.

\begin{vccreduction}[Crown Removal Reduction~\cite{strash-2022}]
\label{red:crown}
Let $(H,I)$ be a head and crown with matching $M$ and unmatched vertices $I'\subseteq I$. Then add cliques in $M$ and $I'$ to the clique cover and remove $N[I]$ from the graph. (See Figure~\ref{fig:crown1}.)
\end{vccreduction}

Note that it is possible to identify flared crowns by applying a reduction based on an LP relaxation, originally introduced for the minimum vertex cover problem by Nemhauser and Trotter~\cite{nemhauser-1975}. A variant of this algorithm due to Iwata et al.~\cite{iwata-2014} identifies and removes \emph{all} flared crowns at once by computing a maximum matching on a bipartite graph with $2n$ vertices and $2m$ edges using the Hopcroft-Karp algorithm~\cite{hopcroft-karp} with running time $O(m\sqrt{n})$.

As Figure~\ref{fig:crown2} illustrates, after exhaustively applying Gramm et al.'s~\cite{gramm-2009} ECC reductions it is possible to have a crown structure after transforming to $G'_{VCC}$. Thus, lifting the crown removal reduction can further reduce an ECC instance.
However, algorithms for computing a maximum matching for the LP relaxation use an explicit representation of $G'_{VCC}$ and therefore it is unclear how to run this reduction without first transforming $G$ to $G'_{VCC}$. The transformation and maximum matching can be computed in time $O(d^2 m + d^2m\sqrt{m}) = O(d^2m^{3/2})$, since there are $O(m)$ vertices and $O(d^2m)$ edges in $G'_{VCC}$. We leave the question of whether the LP relaxation reduction can be more efficiently lifted to an ECC reduction as an open problem.

\subsubsection{Folding-Based VCC Reductions}
In contrast to clique-removal-based reductions, \emph{folding-based} reductions contract a subset $S\subseteq V$ of vertices into a single vertex $v'$. \emph{Folding} $S$ produces a new graph $G^f = (V^f, E^f)$ with $V^f = (V\setminus S) \cup \{v'\}$ and $E^f=(E\setminus\{\{v,x\}\in E\mid v\in S\})\cup\{\{v',x\}\mid \exists v\in S, x\not\in S, \{v,x\}\in E\}$. We discuss the connections between the ECC problem and the simplest folding-based reduction, folding vertices of degree two.
\paragraph*{Degree-2 Folding}

\begin{figure}[!t]
\begin{subfigure}[t]{0.48\textwidth}
\centering
\includegraphics[]{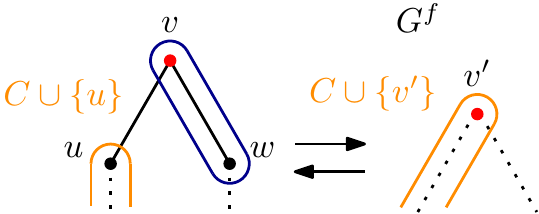}
\caption{The degree-2 folding VCC reduction.}
\label{fig:degreetwo1}
\end{subfigure}
\begin{subfigure}[t]{0.48\textwidth}
\centering
\includegraphics[]{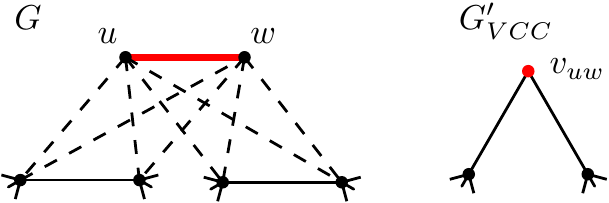}
\caption{Edge $\{u,w\}\in E'$ in $G$ transforms into a degree-2 vertex $v_{uw}$ with non-adjacent neighbors in $G'_{VCC}$.}
\label{fig:degreetwo2}
\end{subfigure}
\caption{The degree-2 folding VCC reduction can be applied after transforming $G$ to $G'_{VCC}$.}
\end{figure}

The degree-2 folding reduction for VCC contracts a degree-2 vertex $v$ with non-adjacent neighbors $u$ and $w$ that are \emph{crossing independent}~\cite{strash-2022}. That is, for each edge $\{x,y\} \subseteq N(u) \cup N(w)$ either $\{x,y\}\subseteq N(u)$ or $\{x,y\}\subseteq N(w)$. This condition ensures that no spurious cliques are formed after folding. A vertex $v$ meeting these conditions is \emph{foldable}.

\begin{vccreduction} [Degree-2 Folding Reduction~\cite{strash-2022}] 
\label{vcc:degree2}
Let $v \in V$ be a foldable degree-2 vertex with non-adjacent neighbors $N(v) = \{u, w\}$. Let $G^f$ be the graph obtained by folding $\{v,u,w\}$. Let $\mathcal{C}^f$ be a minimum VCC of $G^f$ with clique $C_{v'} \in \mathcal{C}^f$ covering vertex $v'$ and let $C = C_{v'} \setminus \{v'\}$. Then, the clique cover
\[\mathcal{C} = \begin{cases} 
    (\mathcal{C}^f \setminus \{C_{v'}\}) \cup \{C \cup \{u\}, \{v, w\}\} & \mbox{if } C \subseteq N(u)\mbox{,} \\
    (\mathcal{C}^f \setminus \{C_{v'}\}) \cup \{C \cup \{w\}, \{v, u\}\} & \mbox{otherwise,}
\end{cases}\]
is a minimum VCC of $G$.
\end{vccreduction}

See Figure~\ref{fig:degreetwo1} for an example of the degree-2 VCC reduction. We note that the transformation from an ECC instance to a VCC instance by Kou et al.~\cite{10.1145/359340.359346} does not produce any degree-2 vertices with non-adjacent neighbors, as edges forming a triangle or 4-clique in $G$ form a triangle or 6-clique in $G_{VCC}$.
However, our transformation with covered edges can result in such vertices (see Figure~\ref{fig:degreetwo2}). Thus, the degree-2 folding VCC reduction can be used to further reduce the instance when applied to $G'_{VCC}$.

We leave as an open problem whether folding-based rules can be lifted to new ECC reductions; we conjecture that it is possible to lift at least degree-2 folding.
However, given how effective the degree-2 folding reduction is in practice for the VCC problem, we highly recommend applying it, even though it incurs the overhead of the transformation to $G'_{VCC}$.

\subsection{Wrapping It All Up} With the tools in this section in hand, we have a clear path to solving the ECC problem on sparse graphs: first apply the data reductions due to Gramm et al.~\cite{gramm-2009}, then transform the partially-covered graph into a VCC instance, which can then be reduced further and solved with any VCC solver. We next perform experiments to evaluate this method.

\section{Experimental Evaluation}
We now compare our technique to the state of the art through extensive experiments on both synthetic instances and real-world graphs.

\subsection{Experimental Setup}
     We implemented the ECC reductions and ECC to VCC transformation in C++ and integrated our methods with the VCC reductions and VCC algorithms by Strash and Thompson\footnote{\url{https://github.com/darrenstrash/ReduVCC}}~\cite{strash-2022}, which we then compiled with \texttt{g++} version 11 using the \texttt{-O3} optimization flag. Our source code will be made available under the open source MIT license.
    All experiments were conducted on Hamilton College's High Performance Computing Cluster (HPCC), on a machine running CentOS Linux 7.8.2003, with four Intel Xeon Gold 6248 processors running at 2.50GHz with 20 cores each, and 1.5TB of memory. Each algorithm is run sequentially on its own core.

    We run experiments on six different algorithms. \textsf{Gramm} is the original branch-and-reduce code by Gramm et al.~\cite{gramm-2009} written in OCAML, which we compiled with \texttt{ocamlc} version 3.10.2, and provided a sufficiently large stack size due to its heavy use of recursion. We implement three algorithms in C++ that first exhaustively apply ECC data reductions, perform a problem reduction to a VCC instance, apply VCC reductions, and then run a VCC solver: \textsf{Redu$^3$BnR} solves with the VCC branch-and-reduce algorithm by Strash and Thompson~\cite{strash-2022}, \textsf{Redu$^3$IG} solves with the VCC iterated greedy (IG) heuristic algorithm by Chalupa~\cite{chalupa2016const}, and \textsf{Redu$^3$ILP} solves with an assignment-based ILP formulation~\cite{DBLP:conf/latin/JabrayilovM18,mehrotra-1996} for VCC and Gurobi version 9.5.1. Finally, the two heuristic algorithms \textsf{Conte}~\cite{conte-2020} and \textsf{EO-ECC}~\cite{abdullah-2022} are from their respective authors and are compiled with \texttt{javac} version 8 and \texttt{g++} version 11 with \texttt{-O3}, respectively. Unless stated otherwise, we run each algorithm with a 24-hour time limit. Our stated running times do not include I/O time such as graph reading and writing. 

In our tables, `Kernel' denotes the relevant size of the graph after reductions as either uncovered edges (\textsf{Gramm}) or vertices remaining (for VCC-based algorithms). `Time' is the time (in seconds) the solver takes to exactly solve the instance. A `--' indicates that the solver did not finish in the 24-hour time limit. \textbf{Bold} values indicate the value is the smallest among all algorithms in the table. 

We run our experiments on randomly-generated instances as well as real-world graphs.

\textbf{Erd\H{o}s-R\'enyi Graphs.}
We generate 70 instances of varying density using the $G(n,p)$ model of generating an $n$-vertex graph where each edge is selected independently with probability $p$. We use values of $n$ that are powers of two from $64$ to $2048$, with two different values of $p$ for each to show the effect of density on the tested algorithms. We generate 5 graphs with each $n$, $p$ pair using different random seeds to observe the behavior of algorithms on multiple instances of similar size and density. (See Tables~\ref{tab:er-full2} and~\ref{tab:er-full} in Appendix~\ref{app:full} for the full statistics.)

\textbf{Real-World Instances.} We run our experiments on 52 large, sparse, complex networks from the Stanford Network Data Repository (SNAP)\footnote{\url{https://snap.stanford.edu/data/}}, the Laboratory for Web Algorithmics (LAW)\footnote{\url{http://law.di.unimi.it/datasets.php}}, and the Koblenz Network Collection (KONECT)\footnote{\url{http://konect.cc/}}. These graphs include citation networks, web-crawl graphs, and social networks; the largest graph has 18M vertices, and most graphs follow a scale-free degree distribution: there are many low degree vertices and few high degree vertices. The number of vertices and edges for each instance can be found with experimental results in Tables~\ref{tab:exact} and \ref{tab:inexact}.

\subsection{Results on Synthetic Instances}
We begin by comparing the performance of \textsf{Gramm} and \textsf{Redu$^3$BnR} on synthetic instances generated with the Erd\H{o}s-R\'enyi $G(n,p)$ model. We present the average kernel size and running time from the execution of \textsf{Gramm} and \textsf{Redu$^3$BnR} on the 5 instances of each pair of $n$ and $p$ in Table~\ref{tab:er-abbv}. (Individual results can be found in Tables~\ref{tab:er-full2} and \ref{tab:er-full} in Appendix~\ref{app:full}.)

Focusing on running time, \textsf{Gramm} and \textsf{Redu$^3$BnR} are equally matched on very sparse graphs, quickly solving many instances in significantly less than one second. Though, as the density increases even slightly, which can be seen when fixing $n$ but increasing $p$, \textsf{Gramm} is no longer able to solve even small instances in a 24-hour time limit. However, on all instances, \textsf{Redu$^3$BnR} easily computes exact solutions. The reason why is clear, on problems that \textsf{Gramm} is unable to solve, the ECC kernel is large (for the highest density instance with $n=64,p=0.200$, even a kernel of average size 50 is too large for \textsf{Gramm} to solve), whereas the VCC kernels for \textsf{Redu$^3$BnR} are significantly smaller in all cases. Indeed, for the densest graphs of each value of $n$, \textsf{Gramm} is unable to solve every instance in 24 hours, but \textsf{Redu$^3$BnR} solves all graphs in less than a second. This illustrates that the combined reduction power of ECC and VCC reductions is able to handle denser instances than running ECC reductions alone.

\begin{table}
\begin{center}
\small
\caption{Results on small Erd\H{o}s-R\'enyi graphs of varying density. 
%with exact branch-and-reduce algorithms \textsf{Gramm} and \textsf{Redu$^3$BnR} on small Erd\H{o}s-R\'enyi graphs in varying density. For each value of $n$ and $p$, we average over 5 instances generated with a different seed. 
A `$^*$' indicates that not all runs finished in the 24-hour time limit, `--' indicates that no runs finished in the 24-hour time limit.}
\label{tab:er-abbv}
\setlength{\tabcolsep}{.7ex}
\begin{tabular}{rrr@{\hskip 3pt} rr@{\hskip 10pt} rr}
\toprule
\multicolumn{3}{c}{Graph}  & \multicolumn{2}{c}{\textsf{Gramm}} & \multicolumn{2}{c}{\textsf{Redu$^3$BnR}}\\
\cmidrule(r){1-3} \cmidrule(r){4-5} \cmidrule(r){6-7}                          
	$n$&	$p$&	$m$	&Kernel & Time (s) & Kernel & Time (s)\\\hline

%\numprint{64} & \nprounddigits{3}\numprint{0.1} & \numprint{98} & \bf\numprint{0}  & \bf$<$ \numprint{0.01} & \bf\numprint{0} & \bf $<$ \numprint{0.01}\\
\numprint{64} & \nprounddigits{3}\numprint{0.15} & \numprint{151} & \numprint{1}  & \bf$<$ \numprint{0.01} & \bf\numprint{0} & \bf$<$ \numprint{0.01} \\
\numprint{64} & \nprounddigits{3}\numprint{0.2} & \numprint{203} & \numprint{50} & \nprounddigits{2}\numprint{1324.52}$^*$\hspace*{-4pt} & \bf\numprint{10} & \bf$<$ \numprint{0.01} \\[3pt]
\numprint{128} & \nprounddigits{3}\numprint{0.1} & \numprint{404} & \bf\numprint{0} & \bf$<$ \numprint{0.01} & \bf\numprint{0} & \bf $<$ \numprint{0.01}\\
\numprint{128} & \nprounddigits{3}\numprint{0.15} & \numprint{610} & \numprint{245} & -- & \bf\numprint{51} & \bf\nprounddigits{2}\numprint{0.02618} \\[3pt]
\numprint{256} & \nprounddigits{3}\numprint{0.075} & \numprint{1217} & \numprint{21} & \nprounddigits{2}\numprint{0.02} & \bf\numprint{0} & \bf$<$ \numprint{0.01} \\
\numprint{256} & \nprounddigits{3}\numprint{0.1} & \numprint{1633} & \numprint{552} & -- & \bf\numprint{12} & \bf\nprounddigits{2}\numprint{0.0218} \\[3pt]
\numprint{512} & \nprounddigits{3}\numprint{0.05} & \numprint{3279} & \numprint{69} & \nprounddigits{2}\numprint{0.206} & \bf\numprint{1} & \bf\nprounddigits{2}\numprint{0.02472} \\
\numprint{512} & \nprounddigits{3}\numprint{0.065} & \numprint{4258} & \numprint{1140} & -- & \bf\numprint{5} & \bf\nprounddigits{2}\numprint{0.03976} \\[3pt]
%\numprint{1024} & \nprounddigits{3}\numprint{0.036} & \numprint{9409} & \numprint{520} & \nprounddigits{2}\numprint{1.256} & \bf\numprint{2} & \bf\nprounddigits{2}\numprint{0.06574} \\
\numprint{1024} & \nprounddigits{3}\numprint{0.0365} & \numprint{9537} & \numprint{629} & \nprounddigits{2}\numprint{153.21}$^*$\hspace*{-4pt} & \bf\numprint{4} & \bf\nprounddigits{2}\numprint{0.07718} \\
\numprint{1024} & \nprounddigits{3}\numprint{0.0375} & \numprint{9799} & \numprint{852} & -- & \bf\numprint{4} & \bf\nprounddigits{2}\numprint{0.06644} \\[3pt]
\numprint{2048} & \nprounddigits{3}\numprint{0.025} & \numprint{26123} & \numprint{1574} & \nprounddigits{2}\numprint{5.178} & \bf\numprint{4} & \bf\nprounddigits{2}\numprint{0.18148} \\
\numprint{2048} & \nprounddigits{3}\numprint{0.0275} & \numprint{28745} & \numprint{3618} & -- & \bf\numprint{5} & \bf\nprounddigits{2}\numprint{0.19704} \\
\bottomrule
\end{tabular}
\end{center}
\end{table}

%\subsection{Size of Transformed Instances}

\subsection{Solving Large Real-World Instances Exactly}

We now see which graphs can be solved exactly by one of three algorithms: \textsf{Gramm}, \textsf{Redu$^3$BnR}, and \textsf{Redu$^3$ILP}. The results are presented in Table~\ref{tab:exact}. \textsf{Gramm} was able to solve 12 of the 27 instances exactly; 10 of these graphs were solved because the kernel had 0 uncovered edges and the other two instances (\verb|ca-CondMat| and \verb|ca-GrQc|) had small kernels of less than 100 uncovered edges. However, \textsf{Gramm} exceeds the 24-hour time limit on the 15 other instances, even those with as few as 176 uncovered edges.

In contrast, \textsf{Redu$^3$BnR} solves 18 of the instances. On all instances, the kernel computed by \textsf{Redu$^3$BnR} was smaller than that of \textsf{Gramm}, the smallest of which is on \verb|zhishi-hudong-int|, which is reduced to 2\% of the size of \textsf{Gramm}'s kernel. With the exception of three instances (\verb|email-EuAll|, \verb|web-NotreDame|, and \verb|web-Stanford|), every instance was reduced to at most 10\% of \textsf{Gramm}'s kernel size.
However, the limitations of branch and reduce for the VCC problem begin to show on these instances. Similar to \textsf{Gramm}, \textsf{Redu$^3$BnR} only finishes within the 24-hour time limit on graphs with kernel size less than 100, and therefore its success is largely due to the reduction of the input instance (a pattern observed in other problems~\cite{strash-power-2016}). On the other hand, the Gurobi solver with an ILP formulation is able to solve kernels of much larger size, even up to \numprint{536196} vertices (in the case of \verb|eu-2005|).

\begin{table*}
\small
\caption{Comparing exact algorithms \textsf{Gramm}, \textsf{Redu$^3$BnR}, and \textsf{Redu$^3$ILP} on real-world instances solved by at least one of the algorithms in a 24-hour time limit. Times marked with a `*' indicate that the algorithm's speed was due to programming language differences and not algorithmic improvements.}
\label{tab:exact}
\begin{center}
\setlength{\tabcolsep}{.7ex}
\begin{tabular}{lrr@{\hskip 3pt} rr@{\hskip 10pt} rr@{\hskip 10pt} r}
\toprule
\multicolumn{3}{c}{Graph}  & \multicolumn{2}{c}{\textsf{Gramm}} & \multicolumn{2}{c}{\textsf{Redu$^3$BnR}} & \textsf{Redu$^3$ILP} \\
\cmidrule(r){1-3} \cmidrule(r){4-5} \cmidrule(r){6-7} \cmidrule(r){8-8}
	Name&	$n$&	$m$	&Kernel & Time (s) & Kernel & Time (s) & Time (s)\\\hline

 \verb|ca-AstroPh| & \numprint{18772} & \numprint{198050} & \numprint{2837} & -- & \bf\numprint{0} & \bf\nprounddigits{2}\numprint{0.3251} & \bf\nprounddigits{2}\numprint{0.3251}\\
 \verb|ca-CondMat| & \numprint{23133} & \numprint{93439} & \numprint{62} & \nprounddigits{2}\numprint{1.74} & \bf\numprint{0} & \bf\nprounddigits{2}\numprint{0.1003} & \bf\nprounddigits{2}\numprint{0.1003}\\
 \verb|ca-GrQc| & \numprint{5242} & \numprint{14484} & \numprint{9} & \nprounddigits{2}\numprint{0.15} & \bf\numprint{0} & \bf\nprounddigits{2}\numprint{0.0171} & \bf\nprounddigits{2}\numprint{0.0171}\\
 \verb|ca-HepPh| & \numprint{12008} & \numprint{118489} & \numprint{491} & -- & \bf\numprint{0} & \bf\nprounddigits{2}\numprint{0.1585} & \bf\nprounddigits{2}\numprint{0.1585}\\
 \verb|ca-HepTh| & \numprint{9877} & \numprint{25973} & \numprint{176} & -- & \bf\numprint{0} & \bf\nprounddigits{2}\numprint{0.0316} & \bf\nprounddigits{2}\numprint{0.0316}\\
 \verb|cnr-2000| & \numprint{325557} & \numprint{2738969} & \numprint{755617} & -- & \bf\numprint{23880} & -- & \bf\nprounddigits{2}\numprint{10727.2888}\\
 \verb|dblp-2010| & \numprint{326186} & \numprint{807700} & \numprint{868} & -- & \bf\numprint{0} & \bf\nprounddigits{2}\numprint{1.9773} & \bf\nprounddigits{2}\numprint{1.9773}\\
 \verb|dblp-2011| & \numprint{986324} & \numprint{3353618} & \numprint{8898} & -- & \bf\numprint{50} & \bf\nprounddigits{2}\numprint{9.1293} & \nprounddigits{2}\numprint{9.8761}\\
 \verb|email-EuAll| & \numprint{265214} & \numprint{364481} & \numprint{20648} & -- & \bf\numprint{5064} & -- & \bf\nprounddigits{2}\numprint{6.9919}\\
 \verb|eu-2005| & \numprint{862664} & \numprint{16138468} & \numprint{5555826} & -- & \bf\numprint{536209} & -- & \bf\nprounddigits{2}\numprint{12966.5851}\\
 \verb|p2p-Gnutella04| & \numprint{10876} & \numprint{39994} & \bf\numprint{0} & \nprounddigits{2}\numprint{0.34} & \bf\numprint{0} & \nprounddigits{2}\numprint{0.0549}$^*$\hspace*{-4pt} & \nprounddigits{2}\numprint{0.0549}$^*$\hspace*{-4pt}\\
 \verb|p2p-Gnutella05| & \numprint{8846} & \numprint{31839} & \bf\numprint{0} & \nprounddigits{2}\numprint{0.23} & \bf\numprint{0} & \nprounddigits{2}\numprint{0.0461}$^*$\hspace*{-4pt} & \nprounddigits{2}\numprint{0.0461}$^*$\hspace*{-4pt}\\
 \verb|p2p-Gnutella06| & \numprint{8717} & \numprint{31525} & \bf\numprint{0} & \nprounddigits{2}\numprint{0.33} & \bf\numprint{0} & \nprounddigits{2}\numprint{0.0431}$^*$\hspace*{-4pt} & \nprounddigits{2}\numprint{0.0431}$^*$\hspace*{-4pt}\\
 \verb|p2p-Gnutella08| & \numprint{6301} & \numprint{20777} & \numprint{261} & -- & \bf\numprint{17} & \bf\nprounddigits{2}\numprint{0.0377} & \bf\nprounddigits{2}\numprint{0.0626}\\
 \verb|p2p-Gnutella09| & \numprint{8114} & \numprint{26013} & \numprint{214} & -- & \bf\numprint{5} & \bf\nprounddigits{2}\numprint{0.0432} & \bf\nprounddigits{2}\numprint{0.0762}\\
 \verb|p2p-Gnutella24| & \numprint{26518} & \numprint{65369} & \bf\numprint{0} & \nprounddigits{2}\numprint{0.91} & \bf\numprint{0} & \nprounddigits{2}\numprint{0.0960}$^*$\hspace*{-4pt} & \nprounddigits{2}\numprint{0.0960}$^*$\hspace*{-4pt}\\
 \verb|p2p-Gnutella25| & \numprint{22687} & \numprint{54705} & \bf\numprint{0} & \nprounddigits{2}\numprint{0.63} & \bf\numprint{0} & \nprounddigits{2}\numprint{0.0755}$^*$\hspace*{-4pt} & \nprounddigits{2}\numprint{0.0755}$^*$\hspace*{-4pt}\\
 \verb|p2p-Gnutella30| & \numprint{36682} & \numprint{88328} & \bf\numprint{0} & \nprounddigits{2}\numprint{1.27} & \bf\numprint{0} & \nprounddigits{2}\numprint{0.0926}$^*$\hspace*{-4pt} & \nprounddigits{2}\numprint{0.0926}$^*$\hspace*{-4pt}\\
 \verb|p2p-Gnutella31| & \numprint{62586} & \numprint{147892} & \bf\numprint{0} & \nprounddigits{2}\numprint{2.14} & \bf\numprint{0} & \nprounddigits{2}\numprint{0.2316}$^*$\hspace*{-4pt} & \nprounddigits{2}\numprint{0.2316}$^*$\hspace*{-4pt}\\
 \verb|roadNet-CA| & \numprint{1965206} & \numprint{2766607} & \bf\numprint{0} & \nprounddigits{2}\numprint{115.17} & \bf\numprint{0} & \nprounddigits{2}\numprint{5.6046}$^*$\hspace*{-4pt} & \nprounddigits{2}\numprint{5.6046}$^*$\hspace*{-4pt}\\
 \verb|roadNet-PA| & \numprint{1088092} & \numprint{1541898} & \bf\numprint{0} & \nprounddigits{2}\numprint{45.75} & \bf\numprint{0} & \nprounddigits{2}\numprint{2.9422}$^*$\hspace*{-4pt} & \nprounddigits{2}\numprint{2.9422}$^*$\hspace*{-4pt}\\
 \verb|roadNet-TX| & \numprint{1379917} & \numprint{1921660} & \bf\numprint{0} & \nprounddigits{2}\numprint{73.21} & \bf\numprint{0} & \nprounddigits{2}\numprint{3.6423}$^*$\hspace*{-4pt} & \nprounddigits{2}\numprint{3.6423}$^*$\hspace*{-4pt}\\
 \verb|web-BerkStan| & \numprint{685230} & \numprint{6649470} & \numprint{2096936} & -- & \bf\numprint{152581} & -- & \bf\nprounddigits{2}\numprint{6753.2748}\\
 \verb|web-Google| & \numprint{875713} & \numprint{4322051} & \numprint{266455} & -- & \bf\numprint{16440} & -- & \bf\nprounddigits{2}\numprint{35.5795}\\
 \verb|web-NotreDame| & \numprint{325729} & \numprint{1090108} & \numprint{98861} & -- & \bf\numprint{14553} & -- & \bf\nprounddigits{2}\numprint{20.1006}\\
 \verb|web-Stanford| & \numprint{281903} & \numprint{1992636} & \numprint{523480} & -- & \bf\numprint{57463} & -- & \bf\nprounddigits{2}\numprint{981.8175}\\
 \verb|zhishi-hudong-int| & \numprint{1984484} & \numprint{14428382} & \numprint{1175068} & -- & \bf\numprint{26536} & -- & \bf\nprounddigits{2}\numprint{568.2578}\\

\bottomrule
\end{tabular}
\end{center}
\end{table*}

\subsection{Solving Remaining Instances Heuristically}
We now look at the instances that could not be solved in the 24-hour time limit by any exact method. The results are presented in Table~\ref{tab:inexact}. Nine instances were reduced to VCC within the time limit of 24 hours, and the remaining instances were too large to finish in the time limit (not in the table). After fully transforming the input ECC instance to a reduced VCC instance, we ran the iterated greedy approach \textsf{IG} due to Chalupa et al.~\cite{chalupa2016const}, which we call \textsf{Redu$^3$IG}, and compare its best solution with a lower bound from \textsf{KaMIS}, a state-of-the-art evolutionary algorithm for finding near-maximum independent sets on huge networks~\cite{lamm2017finding}. Four instances were solved to within 300 vertices of optimum, two of which (\verb|soc-Slashdot0811| and \verb|soc-Slashdot0902|) are within 100 vertices. The remaining instances are solved to within \numprint{6000} vertices of optimum. %, including \verb|soc-pokec-relationships|, with the largest cover size of $\approx$12M cliques.

\begin{table*}
\small
\caption{Heuristic solutions for graphs that could not be solved exactly in 24 hours. 
%heuristic algorithm \textsf{Redu$^3$IG} with a 6-hour time limit. 
`lb' is a lower bound on $\theta_E(G)$ from \textsf{KaMIS}, `ub' is the smallest clique cover computed by \textsf{Redu$^3$IG}, and `Time' is the time in seconds for \textsf{Redu$^3$IG} to reach this result.}
\label{tab:inexact}
\begin{center}
\setlength{\tabcolsep}{.7ex}
\begin{tabular}{lrr@{\hskip 3pt} r@{\hskip 10pt} rr@{\hskip 10pt} r}
\toprule
\multicolumn{3}{c}{Graph}  & \multicolumn{1}{c}{\textsf{KaMIS}} & \multicolumn{2}{c}{\textsf{Redu$^3$IG}}\\
\cmidrule(r){1-3} \cmidrule(r){4-4} \cmidrule(r){5-6}
	Name&	$n$&	$m$	&lb & ub & Time (s)\\\hline
\verb|as-skitter| & \numprint{1696415} & \numprint{11095298} & \numprint{5843072} & \numprint{5847591} & \nprounddigits{2}\numprint{20848.1689} \\
\verb|email-Enron| & \numprint{36692} & \numprint{183831} & \numprint{42141} & \numprint{42207} & \nprounddigits{2}\numprint{2200.9986} \\
\verb|soc-Epinions1| & \numprint{75879} & \numprint{405740} & \numprint{185544} & \numprint{186384} & \nprounddigits{2}\numprint{18064.7851} \\
\verb|soc-pokec-relationships| & \numprint{1632803} & \numprint{22301964} & \numprint{12222248} & \numprint{12227949} & \nprounddigits{2}\numprint{21451.9118} \\
\verb|soc-Slashdot0811| & \numprint{77360} & \numprint{469180} & \numprint{328018} & \numprint{328079} & \nprounddigits{2}\numprint{3073.7503} \\
\verb|soc-Slashdot0902| & \numprint{82168} & \numprint{504230} & \numprint{351012} & \numprint{351072} & \nprounddigits{2}\numprint{3125.2094} \\
\verb|wiki-Talk| & \numprint{2394385} & \numprint{4659565} & \numprint{3645692} & \numprint{3648312} & \nprounddigits{2}\numprint{21088.5298} \\
\verb|wiki-Vote| & \numprint{7115} & \numprint{100762} & \numprint{34789} & \numprint{35004} & \nprounddigits{2}\numprint{21424.4757} \\
\verb|zhishi-baidu-relatedpages| & \numprint{415641} & \numprint{2374044} & \numprint{1372941} & \numprint{1373912} & \nprounddigits{2}\numprint{9988.9979} \\
\bottomrule
\end{tabular}
\end{center}
\end{table*}

\subsection{Summarizing the Quality of Existing Heuristic Solvers}
Finally, using our exact results, we evaluate the quality of two heuristic solvers designed for large sparse graphs. We compare \textsf{Conte}, an algorithm by Conte et al.~\cite{conte-2020} designed for large sparse graphs and \textsf{EO-ECC} by Abdullah et al.~\cite{abdullah-2022}.
We run \textsf{Conte} and \textsf{EO-ECC} on all instances that were solved exactly (i.e., those from Table~\ref{tab:exact}). The results are presented in Table~\ref{tab:verify}.

From among the 27 graphs, \textsf{Conte} solves five instances exactly. A further nine instances are solved within 50 cliques of optimal, and eight additional graphs are solved within \numprint{2000} of optimal. \textsf{EO-ECC}, on the other hand, solves eight instances exactly (a superset of \textsf{Conte}'s five) and solves these faster than \textsf{Conte}. Furthermore, \textsf{EO-ECC} finds 14 smaller solutions faster than \textsf{Conte} (\textsf{Conte} only finds four smaller solutions faster). However, a distinct negative is \textsf{EO-ECC}'s running time and solution quality on \verb|cnr-2000|, \verb|eu-2005|, and \verb|web-BerkStan|, which is much worse than \textsf{Conte}. We conclude that \textsf{Conte} gives consistently fast results with reasonable solutions, and \textsf{EO-ECC} is sometimes very fast and accurate, and other times not. 
%seems very well suited for solving road network instances and very sparse graphs. On the larger and slightly more dense instances, such as \verb|cnr-2000|, \verb|dblp-2011|, \verb|eu-2005|, and \verb|zhishi-hudong-int|, the gap distance from optimal is much higher.

\begin{table}[!h]
\small
\caption{Evaluation of the quality of heuristic solvers \textsf{Conte} and \textsf{EO-ECC} on all graphs with known edge clique cover number $\theta_E(G)$. `ub' is the solution found by the given algorithm, and `Time' is the algorithm's time in seconds. Values of `ub' marked in \textbf{bold} indicates the algorithm found an optimal solution, with its time in \textbf{bold} if it did so faster than its competitor. Values of `ub' in \emph{italics} indicate that an algorithm found an ECC smaller than its competitor, with its time in \emph{italics} if it did so faster than its competitor.}
\label{tab:verify}
\begin{center}
\setlength{\tabcolsep}{.7ex}
\begin{tabular}{lrrr@{\hskip 4pt} rr@{\hskip 4pt}rr}
\toprule
\multicolumn{4}{c}{Graph $G$}  & \multicolumn{2}{c}{\textsf{Conte}}& \multicolumn{2}{c}{\textsf{EO-ECC}}\\
\cmidrule(r){1-4} \cmidrule(r){5-6} \cmidrule(r){7-8}
	Name&	$n$&	$m$	&$\theta_E(G)$& ub & Time\,(s)& ub & Time\,(s)\\\hline
\verb|ca-AstroPh| & \numprint{18772} & \numprint{198050} & \numprint{15134} & \numprint{15481} & \nprounddigits{2}\numprint{0.916}  & \it\numprint{15373} & \it\nprounddigits{2}\numprint{0.5} \\
\verb|ca-CondMat| & \numprint{23133} & \numprint{93439} & \numprint{16283} & \numprint{16378} & \nprounddigits{2}\numprint{0.537}  & \it\numprint{16307} & \it\nprounddigits{2}\numprint{0.07} \\
\verb|ca-GrQc| & \numprint{5242} & \numprint{14484} & \numprint{3737} & \numprint{3749} & \nprounddigits{2}\numprint{0.152}  & \it\numprint{3739} & \it\nprounddigits{2}\numprint{0.01} \\
\verb|ca-HepPh| & \numprint{12008} & \numprint{118489} & \numprint{10031} & \numprint{10142} & \nprounddigits{2}\numprint{0.694}  & \it\numprint{10097} & \it\nprounddigits{2}\numprint{0.35} \\
\verb|ca-HepTh| & \numprint{9877} & \numprint{25973} & \numprint{9190} & \numprint{9264} & \nprounddigits{2}\numprint{0.188}  & \it\numprint{9212} & \it\nprounddigits{2}\numprint{0.02} \\
\verb|cnr-2000| & \numprint{325557} & \numprint{2738969} & \numprint{752118} & \it\numprint{756905} & \it\nprounddigits{2}\numprint{14.917}  & \numprint{763365} & \nprounddigits{2}\numprint{2820.97} \\
\verb|dblp-2010| & \numprint{326186} & \numprint{807700} & \numprint{186834} & \numprint{187395} & \nprounddigits{2}\numprint{2.223}  & \it\numprint{186968} & \it\nprounddigits{2}\numprint{0.44} \\
\verb|dblp-2011| & \numprint{986324} & \numprint{3353618} & \numprint{707773} & \numprint{713219} & \nprounddigits{2}\numprint{13.555}  & \it\numprint{709156} & \it\nprounddigits{2}\numprint{3.48} \\
\verb|email-EuAll| & \numprint{265214} & \numprint{364481} & \numprint{297092} & \it\numprint{298943} & \nprounddigits{2}\numprint{2.575}  & \numprint{299257} & \nprounddigits{2}\numprint{2.14} \\
\verb|eu-2005| & \numprint{862664} & \numprint{16138468} & \numprint{2832059} & \numprint{2883585} & \nprounddigits{2}\numprint{108.673}  & \numprint{3032337} & \nprounddigits{2}\numprint{8458.21} \\
\verb|p2p-Gnutella04| & \numprint{10876} & \numprint{39994} & \numprint{38491} & \bf\numprint{38491} & \nprounddigits{2}\numprint{0.291}  & \bf\numprint{38491} & \bf\nprounddigits{2}\numprint{0.04} \\
\verb|p2p-Gnutella05| & \numprint{8846} & \numprint{31839} & \numprint{30523} & \numprint{30527} & \nprounddigits{2}\numprint{0.251}  & \it\numprint{30525} & \it\nprounddigits{2}\numprint{0.04} \\
\verb|p2p-Gnutella06| & \numprint{8717} & \numprint{31525} & \numprint{30322} & \numprint{30327} & \nprounddigits{2}\numprint{0.264}  & \it\numprint{30324} & \it\nprounddigits{2}\numprint{0.04} \\
\verb|p2p-Gnutella08| & \numprint{6301} & \numprint{20777} & \numprint{19000} & \numprint{19042} & \nprounddigits{2}\numprint{0.198}  & \it\numprint{19012} & \it\nprounddigits{2}\numprint{0.03} \\
\verb|p2p-Gnutella09| & \numprint{8114} & \numprint{26013} & \numprint{24117} & \numprint{24150} & \nprounddigits{2}\numprint{0.244}  & \it\numprint{24133} & \it\nprounddigits{2}\numprint{0.03} \\
\verb|p2p-Gnutella24| & \numprint{26518} & \numprint{65369} & \numprint{63725} & \numprint{63726} & \nprounddigits{2}\numprint{0.412}  & \bf\numprint{63725} & \bf\nprounddigits{2}\numprint{0.06} \\
\verb|p2p-Gnutella25| & \numprint{22687} & \numprint{54705} & \numprint{53367} & \bf\numprint{53367} & \nprounddigits{2}\numprint{0.325}  & \bf\numprint{53367} & \bf\nprounddigits{2}\numprint{0.05} \\
\verb|p2p-Gnutella30| & \numprint{36682} & \numprint{88328} & \numprint{85821} & \bf\numprint{85823} & \nprounddigits{2}\numprint{0.515}  & \bf\numprint{85821} & \bf\nprounddigits{2}\numprint{0.1} \\
\verb|p2p-Gnutella31| & \numprint{62586} & \numprint{147892} & \numprint{144478} & \bf\numprint{144478} & \nprounddigits{2}\numprint{0.827}  & \bf\numprint{144478} & \bf\nprounddigits{2}\numprint{0.15} \\
\verb|roadNet-CA| & \numprint{1965206} & \numprint{2766607} & \numprint{2537936} & \numprint{2537945} & \nprounddigits{2}\numprint{17.902}  & \bf\numprint{2537936} & \bf\nprounddigits{2}\numprint{1.02} \\
\verb|roadNet-PA| & \numprint{1088092} & \numprint{1541898} & \numprint{1413370} & \bf\numprint{1413370} & \nprounddigits{2}\numprint{10.618}  & \bf\numprint{1413370} & \bf\nprounddigits{2}\numprint{0.69} \\
\verb|roadNet-TX| & \numprint{1379917} & \numprint{1921660} & \numprint{1763295} & \numprint{1763298} & \nprounddigits{2}\numprint{13.482}  & \bf\numprint{1763295} & \bf\nprounddigits{2}\numprint{0.89} \\
\verb|web-BerkStan| & \numprint{685230} & \numprint{6649470} & \numprint{1834074} & \it\numprint{1850605} & \it\nprounddigits{2}\numprint{54.344}  & \numprint{1903872} & \nprounddigits{2}\numprint{2089.25} \\
\verb|web-Google| & \numprint{875713} & \numprint{4322051} & \numprint{1242770} & \numprint{1254107} & \nprounddigits{2}\numprint{24.961}  & \it\numprint{1251672} & \nprounddigits{2}\numprint{33.1} \\
\verb|web-NotreDame| & \numprint{325729} & \numprint{1090108} & \numprint{451424} & \numprint{453864} & \nprounddigits{2}\numprint{7.085}  & \it\numprint{453805} & \nprounddigits{2}\numprint{7.31} \\
\verb|web-Stanford| & \numprint{281903} & \numprint{1992636} & \numprint{562417} & \it\numprint{570958} & \it\nprounddigits{2}\numprint{16.848}  & \numprint{591957} & \nprounddigits{2}\numprint{326.92} \\
\verb|zhishi-hudong-int| & \numprint{1984484} & \numprint{14428382} & \numprint{10557244} & \numprint{10698424} & \nprounddigits{2}\numprint{123.448}  & \it\numprint{10678121} & \nprounddigits{2}\numprint{322.89} \\
\midrule
\multicolumn{4}{l}{Summary ({\bf\#optimal} / {\it\#smaller and faster})}  & \multicolumn{2}{c}{({\bf 5} / {\it 4})}& \multicolumn{2}{c}{({\bf 8} / {\it 14})}\\
\bottomrule
\end{tabular}
\end{center}
\end{table}

\section{Conclusion and Future Work}
We introduced a technique to further reduce ECC problem instances via VCC data reductions, enabling us to solve sparse real-world graphs that could not be solved before. Critical to this technique is the ability to transform reduced ECC instances to the VCC problem, through a modification of the polynomial-time reduction of Kou et al.~\cite{10.1145/359340.359346}. The combined reduction power of ECC and VCC reductions, which we call \emph{synergistic} data reduction, produces significantly smaller kernels than ECC reductions alone. Of particular interest for future work is integrating data reduction rules with existing heuristic algorithms for the ECC problem, trying to implement a more efficient LP relaxation ECC reduction without a transformation, and to see if folding-based reductions can be lifted to the ECC problem.

\clearpage
\bibliography{ecc}

\clearpage
\appendix

\section{Appendix: Additional ECC Reductions from Gramm et al.}
\label{appendix:gramm}

ECC Reduction~\ref{red:gramm_3} uses the notion of prisoners and exits. For a vertex $v$, the neighbors $p$ with $N(p) \subset N(v)$ are called \emph{prisoners} and the remaining neighbors $x$ with $N(x) \setminus N(v)= \emptyset$ \emph{exits}. We say that the prisoners \emph{dominate} the exits if every exit $x$ has an adjacent prisoner.

\setcounter{eccreduction}{2}
\begin{eccreduction}[\cite{gramm-2009}]
\label{red:gramm_3}
Let $v\in V'$ have at least one prisoner. If each prisoner is adjacent to at least one vertex other than $v$ via an uncovered edge, and every exit has an adjacent prisoner, then delete $v$. To reconstruct a solution for the unreduced instance, add $v$ to every clique containing a prisoner of $v$.
\end{eccreduction}

As noted by Gramm et al.~\cite{gramm-2009}, ECC Reduction~\ref{red:gramm_3} can be applied in $O(n^3)$ by using an edge list representation by testing each vertex $v$ for adjacencies between its prisoners and exits. If using an adjacency list instead, the reduction can be applied to a vertex in $O(\Delta^2)$ time where the maximum degree $\Delta$ is small, reducing to $O(\Delta^2 n)$ overall time. However, a severely limiting restriction is that the rule requires ECC Reduction~\ref{red:gramm_1} and~\ref{red:gramm_2} to have been exhaustively applied first, increasing the overall running time dramatically to $O(m^2)$.

\begin{eccreduction}[\cite{gramm-2009}]
\label{red:gramm_4}
Let $N'(v)$ be the uncovered neighbors of $v$ and $C_1, C_2$, $\ldots$, $C_k$ be the connected components of $G[N'(v)]$. If $k > 1$, then replace $v$ with $v_1, v_2, \ldots, v_k$. For $1\leq i \leq k$, add uncovered edges from $v_i$ to each vertex in $C_i$, and for each covered edge $\{u, v\}$ incident to $v$, add covered edge $\{u, v_i\}$ for all $1\leq i \leq k$.
\end{eccreduction}

ECC Reduction~\ref{red:gramm_4} is expensive and, while it is applied often on denser instances in experiments by Gramm et al.~\cite{gramm-2009}, it is rarely applied on the sparsest instances, and it only moderately improves the running time of branch and reduce.

\clearpage

\section{Appendix: Full Graph Statistics for Erd\H{o}s-R\'enyi Graphs}
\label{app:full}

\small
\begin{table}[!h]
\begin{center}
\caption{Experimental results with exact branch-and-reduce algorithms \textsf{Gramm} and \textsf{Redu$^3$BnR} on small Erd\H{o}s-R\'enyi graphs in varying density.}
\label{tab:er-full2}
\setlength{\tabcolsep}{.7ex}
\begin{tabular}{rrr@{\hskip 3pt} rr@{\hskip 10pt} rr}
\toprule
\multicolumn{3}{c}{Graph}  & \multicolumn{2}{c}{\textsf{Gramm}} & \multicolumn{2}{c}{\textsf{Redu$^3$BnR}}\\
\cmidrule(r){1-3} \cmidrule(r){4-5} \cmidrule(r){6-7}                          
	$n$&	$p$&	$m$	&Kernel & Time (s) & Kernel & Time (s)\\\hline
64 &\nprounddigits{3}\numprint{0.1} & 95 & 0 & $<$\nprounddigits{2}\numprint{0.01} & 0 &$<$\nprounddigits{2}\numprint{0.01} \\
64 &\nprounddigits{3}\numprint{0.1} & 101 & 0 & $<$\nprounddigits{2}\numprint{0.01} & 0 &$<$\nprounddigits{2}\numprint{0.01} \\
64 &\nprounddigits{3}\numprint{0.1} & 103 & 0 & $<$\nprounddigits{2}\numprint{0.01} & 0 &$<$\nprounddigits{2}\numprint{0.01} \\
64 &\nprounddigits{3}\numprint{0.1} & 95 & 0 & $<$ \nprounddigits{2}\numprint{0.01} & 0 &$<$\nprounddigits{2}\numprint{0.01} \\
64 &\nprounddigits{3}\numprint{0.1} & 94 & 0 & $<$\nprounddigits{2}\numprint{0.01} & 0 &$<$\nprounddigits{2}\numprint{0.01} \\
64 &\nprounddigits{3}\numprint{0.15} & 147 & 0 & $<$\nprounddigits{2}\numprint{0.01} & 0 &$<$\nprounddigits{2}\numprint{0.01} \\
64 &\nprounddigits{3}\numprint{0.15} & 151 & 2 & $<$\nprounddigits{2}\numprint{0.01} & 0 &$<$\nprounddigits{2}\numprint{0.01} \\
64 &\nprounddigits{3}\numprint{0.15} & 159 & 0 & $<$\nprounddigits{2}\numprint{0.01} & 0 &$<$\nprounddigits{2}\numprint{0.01} \\
64 &\nprounddigits{3}\numprint{0.15} & 152 & 0 & $<$\nprounddigits{2}\numprint{0.01} & 0 &$<$\nprounddigits{2}\numprint{0.01} \\
64 &\nprounddigits{3}\numprint{0.15} & 145 & 0 & $<$\nprounddigits{2}\numprint{0.01} & 0 &$<$\nprounddigits{2}\numprint{0.01} \\
64 &\nprounddigits{3}\numprint{0.2} & 197 & 27 &\nprounddigits{2}\numprint{2.23} & 0 &$<$\nprounddigits{2}\numprint{0.01} \\
64 &\nprounddigits{3}\numprint{0.2} & 211 & 55 & -- & 5 &$<$\nprounddigits{2}\numprint{0.01} \\
64 &\nprounddigits{3}\numprint{0.2} & 206 & 55 & -- & 41 &$<$\nprounddigits{2}\numprint{0.01} \\
64 &\nprounddigits{3}\numprint{0.2} & 208 & 61 & -- & 0 &$<$\nprounddigits{2}\numprint{0.01} \\
64 &\nprounddigits{3}\numprint{0.2} & 189 & 51 &\nprounddigits{2}\numprint{2646.81} & 0 &$<$\nprounddigits{2}\numprint{0.01} \\
128 &\nprounddigits{3}\numprint{0.1} & 403 & 0 & $<$\nprounddigits{2}\numprint{0.01} & 0 &$<$\nprounddigits{2}\numprint{0.01} \\
128 &\nprounddigits{3}\numprint{0.1} & 407 & 0 & $<$\nprounddigits{2}\numprint{0.01} & 0 &$<$\nprounddigits{2}\numprint{0.01} \\
128 &\nprounddigits{3}\numprint{0.1} & 400 & 0 & $<$\nprounddigits{2}\numprint{0.01} & 0 &$<$\nprounddigits{2}\numprint{0.01} \\
128 &\nprounddigits{3}\numprint{0.1} & 399 & 0 & $<$\nprounddigits{2}\numprint{0.01} & 0 &$<$\nprounddigits{2}\numprint{0.01} \\
128 &\nprounddigits{3}\numprint{0.1} & 407 & 0 &\nprounddigits{2}\numprint{0.01} & 0 &$<$\nprounddigits{2}\numprint{0.01} \\
128 &\nprounddigits{3}\numprint{0.15} & 598 & 224 & -- & 5 &$<$\nprounddigits{2}\numprint{0.01} \\
128 &\nprounddigits{3}\numprint{0.15} & 610 & 255 & -- & 91 &$<$\nprounddigits{2}\numprint{0.01} \\
128 &\nprounddigits{3}\numprint{0.15} & 608 & 223 & -- & 0 &$<$\nprounddigits{2}\numprint{0.01} \\
128 &\nprounddigits{3}\numprint{0.15} & 622 & 295 & -- & 158 &\nprounddigits{2}\numprint{0.1013} \\
128 &\nprounddigits{3}\numprint{0.15} & 609 & 228 & -- & 0 &$<$\nprounddigits{2}\numprint{0.01} \\
256 &\nprounddigits{3}\numprint{0.075} & 1230 & 8 &\nprounddigits{2}\numprint{0.02} & 0 &\nprounddigits{2}\numprint{0.0108} \\
256 &\nprounddigits{3}\numprint{0.075} & 1220 & 52 &\nprounddigits{2}\numprint{0.03} & 0 &\nprounddigits{2}\numprint{0.0102} \\
256 &\nprounddigits{3}\numprint{0.075} & 1208 & 38 &\nprounddigits{2}\numprint{0.02} & 0 &$<$\nprounddigits{2}\numprint{0.01} \\
256 &\nprounddigits{3}\numprint{0.075} & 1207 & 0 &\nprounddigits{2}\numprint{0.02} & 0 &$<$\nprounddigits{2}\numprint{0.01} \\
256 &\nprounddigits{3}\numprint{0.075} & 1220 & 6 &\nprounddigits{2}\numprint{0.01} & 0 &$<$\nprounddigits{2}\numprint{0.01} \\
256 &\nprounddigits{3}\numprint{0.1} & 1640 & 587 & -- & 28 &\nprounddigits{2}\numprint{0.0236} \\
256 &\nprounddigits{3}\numprint{0.1} & 1624 & 522 & -- & 0 &\nprounddigits{2}\numprint{0.0247} \\
256 &\nprounddigits{3}\numprint{0.1} & 1618 & 541 & -- & 24 &\nprounddigits{2}\numprint{0.0209} \\
256 &\nprounddigits{3}\numprint{0.1} & 1636 & 548 & -- & 5 &\nprounddigits{2}\numprint{0.0186} \\
256 &\nprounddigits{3}\numprint{0.1} & 1646 & 561 & -- & 0 &\nprounddigits{2}\numprint{0.0212} \\
\bottomrule
\end{tabular}
\end{center}
\end{table}

\begin{table}
\begin{center}
\caption{Experimental results with exact branch-and-reduce algorithms \textsf{Gramm} and \textsf{Redu$^3$BnR} on small Erd\H{o}s-R\'enyi graphs in varying density.}
\label{tab:er-full}
\setlength{\tabcolsep}{.7ex}
\begin{tabular}{rrr@{\hskip 3pt} rr@{\hskip 10pt} rr}
\toprule
\multicolumn{3}{c}{Graph}  & \multicolumn{2}{c}{\textsf{Gramm}} & \multicolumn{2}{c}{\textsf{Redu$^3$BnR}}\\
\cmidrule(r){1-3} \cmidrule(r){4-5} \cmidrule(r){6-7}                          
$n$& $p$& $m$ &Kernel & Time (s) & Kernel & Time (s)\\\hline
512 &\nprounddigits{3}\numprint{0.05} & 3288 & 107 &\nprounddigits{2}\numprint{0.68} & 0 &\nprounddigits{2}\numprint{0.0312} \\
512 &\nprounddigits{3}\numprint{0.05} & 3274 & 115 &\nprounddigits{2}\numprint{0.16} & 0 &\nprounddigits{2}\numprint{0.0254} \\
512 &\nprounddigits{3}\numprint{0.05} & 3290 & 11 &\nprounddigits{2}\numprint{0.05} & 0 &\nprounddigits{2}\numprint{0.0199} \\
512 &\nprounddigits{3}\numprint{0.05} & 3284 & 69 &\nprounddigits{2}\numprint{0.08} & 0 &\nprounddigits{2}\numprint{0.025} \\
512 &\nprounddigits{3}\numprint{0.05} & 3259 & 39 &\nprounddigits{2}\numprint{0.06} & 5 &\nprounddigits{2}\numprint{0.0221} \\
512 &\nprounddigits{3}\numprint{0.065} & 4279 & 1167 & -- & 0 &\nprounddigits{2}\numprint{0.0399} \\
512 &\nprounddigits{3}\numprint{0.065} & 4268 & 1165 & -- & 15 &\nprounddigits{2}\numprint{0.04} \\
512 &\nprounddigits{3}\numprint{0.065} & 4262 & 1220 & -- & 0 &\nprounddigits{2}\numprint{0.0467} \\
512 &\nprounddigits{3}\numprint{0.065} & 4261 & 1100 & -- & 5 &\nprounddigits{2}\numprint{0.0425} \\
512 &\nprounddigits{3}\numprint{0.065} & 4216 & 1048 & -- & 5 &\nprounddigits{2}\numprint{0.0297} \\
1024 &\nprounddigits{3}\numprint{0.036} & 9410 & 578 &\nprounddigits{2}\numprint{0.71} & 5 &\nprounddigits{2}\numprint{0.0631} \\
1024 &\nprounddigits{3}\numprint{0.036} & 9433 & 525 &\nprounddigits{2}\numprint{0.74} & 0 &\nprounddigits{2}\numprint{0.0564} \\
1024 &\nprounddigits{3}\numprint{0.036} & 9314 & 389 &\nprounddigits{2}\numprint{0.48} & 0 &\nprounddigits{2}\numprint{0.0754} \\
1024 &\nprounddigits{3}\numprint{0.036} & 9520 & 609 &\nprounddigits{2}\numprint{3.28} & 5 &\nprounddigits{2}\numprint{0.0758} \\
1024 &\nprounddigits{3}\numprint{0.036} & 9366 & 496 &\nprounddigits{2}\numprint{1.07} & 0 &\nprounddigits{2}\numprint{0.058} \\
1024 &\nprounddigits{3}\numprint{0.0365} & 9524 & 667 &\nprounddigits{2}\numprint{457.86} & 5 &\nprounddigits{2}\numprint{0.0771} \\
1024 &\nprounddigits{3}\numprint{0.0365} & 9552 & 644 &\nprounddigits{2}\numprint{1.26} & 0 &\nprounddigits{2}\numprint{0.0739} \\
1024 &\nprounddigits{3}\numprint{0.0365} & 9459 & 519 &\nprounddigits{2}\numprint{0.51} & 0 &\nprounddigits{2}\numprint{0.0876} \\
1024 &\nprounddigits{3}\numprint{0.0365} & 9651 & 725 & -- & 15 &\nprounddigits{2}\numprint{0.0739} \\
1024 &\nprounddigits{3}\numprint{0.0365} & 9497 & 589 & -- & 0 &\nprounddigits{2}\numprint{0.0734} \\
1024 &\nprounddigits{3}\numprint{0.0375} & 9776 & 847 & -- & 5 &\nprounddigits{2}\numprint{0.0641} \\
1024 &\nprounddigits{3}\numprint{0.0375} & 9803 & 908 & -- & 0 &\nprounddigits{2}\numprint{0.0556} \\
1024 &\nprounddigits{3}\numprint{0.0375} & 9735 & 765 & -- & 0 &\nprounddigits{2}\numprint{0.0584} \\
1024 &\nprounddigits{3}\numprint{0.0375} & 9929 & 985 & -- & 15 &\nprounddigits{2}\numprint{0.0699} \\
1024 &\nprounddigits{3}\numprint{0.0375} & 9751 & 754 & -- & 0 &\nprounddigits{2}\numprint{0.0842} \\
2048 &\nprounddigits{3}\numprint{0.025} & 26156 & 1591 &\nprounddigits{2}\numprint{6.85} & 15 &\nprounddigits{2}\numprint{0.1898} \\
2048 &\nprounddigits{3}\numprint{0.025} & 26128 & 1588 &\nprounddigits{2}\numprint{5.24} & 0 &\nprounddigits{2}\numprint{0.1741} \\
2048 &\nprounddigits{3}\numprint{0.025} & 26136 & 1618 &\nprounddigits{2}\numprint{5.1} & 5 &\nprounddigits{2}\numprint{0.1663} \\
2048 &\nprounddigits{3}\numprint{0.025} & 26159 & 1469 &\nprounddigits{2}\numprint{3.97} & 0 &\nprounddigits{2}\numprint{0.1961} \\
2048 &\nprounddigits{3}\numprint{0.025} & 26035 & 1600 &\nprounddigits{2}\numprint{4.73} & 0 &\nprounddigits{2}\numprint{0.1811} \\
2048 &\nprounddigits{3}\numprint{0.0275} & 28741 & 3706 & -- & 15 &\nprounddigits{2}\numprint{0.1913} \\
2048 &\nprounddigits{3}\numprint{0.0275} & 28745 & 3474 & -- & 0 &\nprounddigits{2}\numprint{0.1885} \\
2048 &\nprounddigits{3}\numprint{0.0275} & 28807 & 3654 & -- & 0 &\nprounddigits{2}\numprint{0.1943} \\
2048 &\nprounddigits{3}\numprint{0.0275} & 28818 & 3663 & -- & 10 &\nprounddigits{2}\numprint{0.2064} \\
2048 &\nprounddigits{3}\numprint{0.0275} & 28612 & 3590 & -- & 0 &\nprounddigits{2}\numprint{0.2047} \\
\bottomrule
\end{tabular}
\end{center}
\end{table}

\end{document}